\documentclass[journal]{IEEEtran}
\usepackage{amsmath,amsfonts}
\usepackage{algorithmic}
\usepackage{algorithm}
\usepackage{array}
\usepackage[caption=false,font=normalsize,labelfont=sf,textfont=sf]{subfig}
\usepackage{textcomp}
\usepackage{stfloats}
\usepackage{url}
\usepackage{verbatim}
\usepackage{graphicx}
\usepackage{amsmath}
\usepackage{amssymb}
\usepackage{cuted} 
\usepackage{mathrsfs}
\usepackage{makecell}
\usepackage{threeparttable}
\usepackage[figuresright]{rotating}
\usepackage{arydshln}
\usepackage{setspace}
\usepackage{color, xcolor} 
\usepackage{cite}
\newtheorem{lemma}{Lemma}
\newtheorem{theorem}{Theorem}

\newtheorem{definition}{Definition}
\newenvironment{proof}{{\noindent\it Proof:} }{\hfill $\square$\par}

\newtheorem{example}{Example}
\newtheorem{remark}{Remark}
\newenvironment{note}{{\it Note:}}{}
\begin{document}

\title{Symplectic self-orthogonal quasi-cyclic codes}

\author{Chaofeng Guan, Ruihu Li, Jingjie Lv, Zhi Ma~\IEEEmembership{}
}
\markboth{}%
{}

\IEEEpubid{}

\maketitle

\begin{abstract}
In this paper, we establish the necessary and sufficient conditions for quasi-cyclic (QC) codes with index even to be symplectic self-orthogonal. Subsequently, we present the 
lower and upper bounds on the minimum symplectic distances of a class of $1$-generator QC codes and their symplectic dual codes by decomposing code spaces.
As an application, we construct numerous new binary symplectic self-orthogonal QC codes with excellent parameters, leading to $117$ record-breaking quantum error-correction codes. 
\end{abstract}

\begin{IEEEkeywords}
Quasi-cyclic code, symplectic self-orthogonal code, record-breaking quantum code.
\end{IEEEkeywords}

\section{Introduction}
\IEEEPARstart{Q}{uantum} error-correction codes (QECCs) are capable of encoding quantum bits to overcome unwanted noise and decoherence and are considered one of the cornerstones of quantum computing and quantum communication. 
The concept of QECCs was initially proposed by Shor \cite{shors1995scheme} and Steane \cite{steane1996multiple}.
Subsequently, many scholars have studied QECCs and gradually established the connection between QECCs and classical codes with self-orthogonal properties under Euclidean, Hermitian, and symplectic inner products \cite{calderbank1998quantum,ketkar2006nonbinary,steane1999enlargement}.
Later, a large number of QECCs with suitable parameters were constructed using classical codes such as cyclic codes \cite{aly2007quantum,shi2019dual}, constacyclic codes \cite{kai2014constacyclic,chen2015application}, and algebraic geometry codes \cite{jin2011euclidean,bartoli2021certain}.
However, due to the inherent limitations of these classical codes, the constructed QECCs are often limited to specific lengths and dimensions. Therefore, these good QECCs are discrete.
Specifically, for binary QECCs, improving the existing records in Grassl's online code table \cite{Grassltable} faces two significant challenges. 
First, many researchers have tried exhaustive searches when the lengths are small. Without additional clever strategies, it is unlikely to find better codes. Second, for QECCs with long lengths, the computational complexity of determining their actual distances will be very high.
Therefore, more powerful tools are needed to construct QECCs with excellent parameters.

As an effective generalization of cyclic codes, quasi-cyclic (QC) codes possess a rich algebraic structure that allows for the construction of flexible codes in terms of lengths and dimensions. In 1974, Kassmi \cite{kasami1974gilbert} proved that QC codes satisfy the modified Gilbert-Varshamov bound, thus indicating their good asymptotic performance. In recent decades, many scholars have conducted in-depth research on the algebraic structure of QC codes \cite{ling2001algebraic,guneri2012concatenated,semenov2012spectral,ezerman2021comparison}, 
resulting in the development of numerous record-breaking linear codes from QC codes \cite{aydin2001structure,daskalov2003new,Chen2016NewBH,aydin2019equivalence}. These studies have demonstrated the excellent performances of QC codes. Therefore, it is natural to consider using QC codes to construct high-quality QECCs. 
Work on constructing QECCs from QC codes started with Hagiwara et al. \cite{hagiwara2007quantum}, but progress was slow due to the lack of a practical theory. 
In 2018, Galindo et al. \cite{galindo2018quasi} determined some sufficient conditions for dual-containment of a class of QC codes under Euclidean, Hermitian, and symplectic inner products. 
Although they did not construct record-breaking QECCs, their work provided a new perspective.
Inspired by Galindo's work, Ezerman et al. \cite{ezerman2019good} studied QC codes over finite fields of order $4$ and $9$ and obtained some new binary and ternary QECCs.
Lv et al. \cite{lv2020explicit,lv2020constructions} and the authors \cite{guan2022construction,guan2022new} studied sufficient conditions for dual-containing or self-orthogonal of three classes of QC codes with specific algebraic structures, and constructed many record-breaking binary QECCs. These studies demonstrate that QC codes also possess superior properties in constructing QECCs.

Nevertheless, many important issues remain concerning constructing QECCs from QC codes. The crucial issue is identifying what kind of QC codes are self-orthogonal, so one can construct them efficiently. 
Notably, this paper determines, for the first time, the necessary and sufficient conditions for symplectic self-orthogonality of QC codes with index even.
The second problem is determining the bounds on the minimum distance of QC codes. 
Around this problem, this paper proposes lower and upper bounds on the minimum symplectic distance of a class of QC codes with index 2, which improves the results in \cite{dastbasteh2023polynomial}, \cite{GuanAdditive2023} and \cite{galindo2018quasi}. 
As a practical application of our findings, we construct numerous binary QECCs that break previous records in terms of their error-correcting capabilities.

This paper is structured as follows. Section \ref{sec2} provides some basics that we use in this paper.
In Section \ref{sec3}, we discuss the necessary and sufficient conditions for QC codes of index even to be symplectic self-orthogonal.
In Section \ref{sec4}, we provide bounds on the minimum symplectic distances of a class of $1$-generator QC codes and their symplectic dual codes. Then, we give some interesting and record-breaking results.  
Finally, we conclude this paper in Section \ref{sec5}. 
\section{Preliminaries}\label{sec2}
This section presents the basics of classical codes and QECCs to facilitate the development of subsequent sections. For more details, see standard textbooks 
\cite{huffman2010fundamentals,huffman2021concise}.

\subsection{Linear codes}

Let $\mathbb{F}_q$ be a finite field with $q$ elements, where $q$ is a prime power. Let $\mathbb{F}_q^*=\mathbb{F}_q\setminus \{0\}$. 
An $\left[\ell n,k\right]_q$ linear code $\mathscr{C}$ is a $k$-dimensional subspace of  $\mathbb{F}_{q}^{\ell n}$. 
For a codeword $\mathbf{c}=(c_{0},c_1,\ldots, c_{\ell n-1})$ of $\mathscr{C}$, the Hamming weight of $\mathbf{c}$ is $\mathbf{w}_{H}(\mathbf{c})=|\left\{i \mid c_{i} \neq 0, 0 \leq i \leq \ell n -1 \right\}|$.
If $\ell n$ is even and $N=\ell n/2$,
then the symplectic weight of $\mathbf{c}$ is 
$\mathbf{w}_{s}(\mathbf{c})=|\{i \mid (c_{i}, c_{N+i})  \neq(0,0), 0 \leq i \leq N-1 \}|$. 
The support of $\mathbf{c}$, denoted by $\mathrm{Supp}(\mathbf{c})$, is the set of coordinates of nonzero entries in $\mathbf{c}$, which is a subset of $\{0,1,\ldots,\ell n-1\}$.
If the minimum Hamming (resp. symplectic) weight of $\mathscr{C}$ is $d_{H}=\min \left\{\mathbf{w}_{H}(\mathbf{c}) \mid \mathbf{c} \in \mathscr{C}, \mathbf{c} \ne \mathbf{0} \right\}$ (resp. $d_{s}=\min \left\{\mathbf{w}_{s}(\mathbf{c}) \mid \mathbf{c} \in \mathscr{C}, \mathbf{c} \ne \mathbf{0}\right\}$), then $\mathscr{C}$ can be written as an $\left[\ell n,k,d_H\right]_q$ (resp. $\left[\ell n,k,d_s\right]_q^s$ ) code.

For $\mathbf{u }=(u_{0},u_1,  \ldots, u_{\ell n-1})$, 
$\mathbf{v}=(v_{0},v_1,\ldots, v_{\ell n-1}) \in \mathbb{F}_q^{\ell n}$,  the Euclidean and symplectic inner products of them are  $\left\langle\mathbf{u}, \mathbf{v} \right\rangle_{e}=\sum\limits_{i=0}^{\ell n-1} u_{i} v_{i}$ and 
$\left\langle\mathbf{u},\mathbf{v} \right\rangle_{s}=\sum\limits_{i=0}^{N-1}\left(u_{i} v_{N+i}-u_{N+i} v_{i}\right)$, respectively. 
The Euclidean and symplectic dual codes of $\mathscr{C}$ are  $\mathscr{C}^{\perp_{e}}=\left\{\mathbf{v} \in \mathbb{F}_{q}^{\ell n} \mid\left\langle\mathbf{u}, \mathbf{v} \right\rangle_{e}=0, \forall \mathbf{u} \in \mathscr{C}\right\}$ and
$ \mathscr{C}^{\perp_{s}}= \left\{\mathbf{v} \in \mathbb{F}_{q}^{\ell n} \mid\left\langle\mathbf{u},\mathbf{v} \right\rangle_{s}=0, \forall \mathbf{u} \in \mathscr{C}\right\}.$
If $\mathscr{C} \subset \mathscr{C}^{\perp_{e}}$ (resp. $\mathscr{C} \subset \mathscr{C}^{\perp_{s}}$), then $\mathscr{C}$ is a Euclidean (resp. symplectic) self-orthogonal code.

For $i \in\{1, 2\}$, let $\mathscr{C}_i$ be an $\left[n_i , k_i , d_i \right]_q$ code. Then, the \textit{ direct sum}  of $\mathscr{C}_1$ and $\mathscr{C}_2$ is an $\left[n_1 + n_2, k_1 + k_2, \min\{d_1, d_2\}\right]_q$ code, denoted by $\mathscr{C}_1\oplus  \mathscr{C}_2$.
The following lemma helps identify symplectic self-orthogonal codes.

\begin{lemma}\label{lemma-self-orthogonal}(\cite{huffman2021concise})
	Let $\ell$ be an even integer and $\mathscr{C}$ be a linear code over $\mathbb{F}_q^{\ell n}$. Then, $\mathscr{C}$ is a self-orthogonal code under the symplectic inner product if and only if for all $\mathbf{c_1}, \mathbf{c_2} \in \mathscr{C}$, it holds that $\left\langle \mathbf{c_1}, \mathbf{c_2}\right\rangle_s = 0$.
\end{lemma}

\subsection{Cyclic codes and quasi-cyclic codes}

Let $\left\langle x^{n}-1\right\rangle$ represent the ideal of the polynomial ring $\mathbb{F}_q\left[x\right]$ generated by $x^n -1$ and $\mathbb{R}_{q,n} = \mathbb{F}_q\left[x\right]/\left\langle x^{n}-1\right\rangle$ be the quotient ring of $\mathbb{F}_q\left[x\right]$ modulo $\left\langle x^{n}-1\right\rangle$.
An element in  $\mathbb{R}_{q,n}$ is an equivalence class of polynomials in $\mathbb{F}_q[x]$. Each equivalence class has a representative which is a polynomial in $\mathbb{F}_q[x]$ of degree less than $n$. Here, equivalence classes in $\mathbb{R}_{q,n}$ are represented by this representative. For a polynomial $c(x)$ in $\mathbb{F}_q[x]$, we use $[c(x)]$ to denote its representative.

A linear $\left[n,k\right]_q$ code $\mathscr{C}$ is cyclic if and only if for any codeword $\mathbf{c}=(c_{0}, c_{1}, \ldots, c_{n-1})$ of $\mathscr{C}$, the shifted vector $\tau_1(\mathbf{c})=(c_{n-1}, c_{0}, \ldots, c_{n-2})$ is also a codeword of $\mathscr{C}$. 
Define an isomorphic mapping $\pi$ as follows:  
\begin{equation}  
\pi: \mathbb{R}_{q,n} \rightarrow \mathbb{F}_q^n, \quad [c(x)]=\sum_{i=0}^{n-1} c_{i} x^{i} \mapsto \mathbf{c}=(c_0,c_1,\ldots,c_{n-1}).  
\end{equation}

Then, we have $\pi(xc(x))=\tau_1(\mathbf{c})$. 
Therefore, when studying cyclic codes, $\left[c(x)\right]$ is identified with the coefficient vector $\mathbf{c}$.  
Since a cyclic code is an ideal in the quotient ring $\mathbb{R}_{q,n}$ and $\mathbb{R}_{q,n}$ is principal, every cyclic code can be generated by a single polynomial in $\mathbb{R}_{q,n}$.
For an $\left[n,k\right]_q$ cyclic code $\mathscr{C}$, we call the monic polynomial $g(x)$ with least degree $n-k$ that generates $\mathscr{C}$ as the generator polynomial of $\mathscr{C}$. 
We use $d_H(\mathscr{C})$ or $d(\left[g(x)\right])$ to denote the minimum Hamming distance of $\mathscr{C}$. 
The parity check polynomial of $\mathscr{C}$ is $h(x)=(x^{n}-1)/g(x)$, whose reciprocal polynomial is \( \tilde{h}(x)= x^{\deg(h(x))} h(x^{-1})\).
The Euclidean dual of $\mathscr{C}$ is also a cyclic code, which has generator polynomial $g^{\perp_e}(x)=h_0^{-1} \tilde{h}(x)$, where $h_0$ is the constant term coefficient of $h(x)$. If $g^{\perp_e}(x)\mid g(x)$, then $\mathscr{C}$ is a Euclidean self-orthogonal code.


A linear code $\mathscr{C}$ of length $\ell n$ over $\mathbb{F}_{q}$ is called a QC code of index $\ell$, if for any codeword $\mathbf{c}=\left(c_{0}, c_{1}, \ldots, c_{\ell n-1}\right)$ of $\mathscr{C}$, the $\ell$ shifted vector  
\[  
\tau_{\ell}(\mathbf{c})=(c_{\ell n-\ell}, \ldots, c_{\ell n-1}, c_0, \ldots, c_{\ell n-\ell-1})  
\]  
is also a codeword.  
By \cite{huffman2021concise}, after appropriate permutations, the generator matrix $G$ of an $h$-generator QC code with index $\ell$ can be put into the following form:
\begin{equation}\setlength{\arraycolsep}{1.2pt}
	G=\left(\begin{array}{cccc}
		A_{1,0} & A_{1,1} & \ldots & A_{1, \ell-1} \\
		A_{2,0} & A_{2,1} & \ldots & A_{2, \ell-1} \\
		\vdots & \vdots & \ddots & \vdots \\
		A_{h, 0} & A_{h, 1} & \ldots & A_{h, \ell-1}
	\end{array}\right),
\end{equation}
where  $A_{i, j}$ $(1\le i \le h, 0\le j \le \ell-1)$ is an $n\times n$ circulant matrix. 
Therefore, a QC code of index $\ell$ also can be identified with an $\mathbb{R}_{q,n}$-submodule of $\mathbb{R}_{q,n}^{\ell}$.

%
%
%

\subsection{Quantum error-correction codes}
Since our purpose is to improve the records in \cite{Grassltable}, here we only introduce the basics of binary QECCs. 
Similar to the classical situation, a binary QECC $\mathrm{Q}$ of length $n$ is a $K$-dimensional subspace of $2^n$-dimensional Hilbert space $(\mathbb{C}^{2})^{\otimes n}$, where $\mathbb{C}$ represents complex field and $(\mathbb{C}^{2})^{\otimes n}$ is the $n$-fold tensor power of $\mathbb{C}^2$. If $K = 2^k$, then the parameters of $\mathrm{Q}$ are denoted by $\left[\left[n, k, d\right]\right]_2$, where $d$ is its minimum distance, and $\mathrm{Q}$ can correct up to $\left\lfloor\frac{d-1}{2}\right\rfloor$ qubit errors. 
\begin{lemma}\label{CRSS}(\cite{calderbank1998quantum}) If $\mathscr{C} \subset \mathbb{F}_{2}^{2 n}$ is a symplectic self-orthogonal $\left[2n, k\right]^s_2$ code, then there exists a QECC with parameters $\left[\left[n, n-k, d_s\right]\right]_{2}$, where $d_s$ is the minimum symplectic distance of $\mathscr{C}^{\perp_{s}} \backslash\mathscr{C}$.
\end{lemma}

In general, QECCs can also be derived from existing ones by the following propagation rules, which will help us to get more record-breaking QECCs later.

\begin{lemma}\label{propagation_rules}
	(\cite{calderbank1998quantum}) Suppose there exists an $\left[\left[n, k, d\right]\right]_{2}$ QECC. Then, the following QECCs also exist.\\
	(1) $\left[\left[n, k-1, d\right]\right]_{2}$ for $k\ge 1$;\\
	(2) $\left[\left[n+1, k, d\right]\right]_{2}$ for $k>0$;\\
	(3) $\left[\left[n-1, k+1, d-1\right]\right]_{2}$ for $n\ge2$.
\end{lemma}

\begin{note}
Throughout the paper, we agree on the following notation.
All the indices $\ell$ of QC codes in the subsequent texts will be even and set $m=\ell/ 2$.
For $f(x)=f_{0}+f_{1} x+\cdots+ f_{n-1} x^{n-1} \in \mathbb{R}_{q,n}$, denote $\bar{f}(x)=x^{n}f(x^{-1})$. 
Let $\mathbf{0}$ and $\mathbf{0}_n$ denote a appropriate zero matrix and a zero vector of length $n$, respectively.
Additionally, we denote 
$\Omega _{mn}=\left(\begin{array}{cc}
	\mathbf{0} & I_{m n} \\
	-I_{m n} & \mathbf{0}
\end{array}\right)$, where $I_{mn}$ is the identity matrix of order $mn$. 
\end{note}

\section{The necessary and sufficient conditions for QC codes with index even to be symplectic self-orthogonal}\label{sec3}


In this section, we present the necessary and sufficient conditions for QC codes with index even to be symplectic self-orthogonal. 
First of all, the following lemma is essential for determining the Euclidean inner product between different polynomials in the form of coefficient vectors.

\begin{lemma} \label{commutative_law}(\cite{galindo2018quasi})
	Let $f_a(x)$, $f_b(x)$, and $f_c(x)$ be polynomials in $\mathbb{R}_{q,n}$. Then the following equation holds for the Euclidean inner product among them 
	\begin{equation}
		\left\langle\left[f_a(x) f_b(x)\right],\left[f_c(x)\right]\right\rangle_{e}=\left\langle\left[f_a(x)\right],\left[\overline{f}_b(x) f_c(x)\right]\right\rangle_{e}.
	\end{equation}
\end{lemma}

Since the symplectic inner product can be decomposed into Euclidean inner product, here we determine the Euclidean orthogonality of cyclic codes before discussing the symplectic orthogonality of QC codes. Under the property of cyclic codes, it is easy to derive Theorem \ref{cornerstone}, which establishes necessary and sufficient conditions for Euclidean orthogonality between two cyclic codes.

\begin{theorem}\label{cornerstone}
	Let $g_i(x)$ be a polynomial in $\mathbb{R}_{q,n}$ that divides $x^{n}-1$, and let $\mathscr{C}_i$ be a cyclic code generated by $\left[g_i(x)\right]$, where $i=1$, $2$.
	Then $\mathscr{C}_1$ and $\mathscr{C}_2$ are orthogonal with each other, i.e., for any $a(x)$ and $b(x)$ in $\mathbb{R}_{q,n}$, 
	$\left\langle\left[a(x) g_1(x)\right],\left[b(x) g_2(x)\right]\right\rangle_{e}=0$ holds if and only if $g_1^{\perp_{e}}(x) \mid g_2(x)$.
\end{theorem}


\begin{definition}\label{one_quasi-cyclic_def}
	Let $g(x)$ be a polynomial in $\mathbb{R}_{q,n}$ that divides ${{x}^{n}}-1$, and let $f_{j}(x)$ be $\ell$ different polynomials in $\mathbb{R}_{q,n}$ for $0\le j\le \ell-1$, such that the greatest common divisor of $f_0(x),f_1(x),\ldots,f_{\ell-1}(x)$ and $\frac{x^n-1}{g(x)}$ is equal to $1$. 
	If $\mathscr{C}$ is a QC code generated by $(\left[g(x){f}_{0}(x)\right]$, $\left[g(x){f}_{1}(x)\right],\ldots,\left[g(x){f}_{\ell -1}(x)\right]) \in \mathbb{R}_{q,n}^{\ell} $, then $\mathscr{C}$ is called a $1$-generator QC code with index $\ell$.
	The generator matrix $G$ of $\mathscr{C}$  has the following form:
	\begin{equation}
		G=\left(G_{0}, G_{1}, \ldots, G_{\ell-1}\right) ,
	\end{equation}
	where $G_ {j} $ is an $n\times n$  circulant matrix generated by $\left[g (x) f_{j} (x)\right]$.
\end{definition}

According to \cite{aydin2001structure}, the dimension of $\mathscr{C}$ in Definition \ref{one_quasi-cyclic_def} is given by $n-\deg(g (x))$.
We next present a theorem that determines the necessary and sufficient conditions for $1$-generator QC codes with index even to be symplectic self-orthogonal.

\begin{theorem}\label{one_quasi-cyclic}
	If $\mathscr{C}$ is a $1$-generator QC code in Definition \ref{one_quasi-cyclic_def} with index even $\ell$, then $\mathscr{C}$ is symplectic self-orthogonal if and only if the following equation holds:
	\begin{equation}
		g^{\perp_e}(x)\mid g(x)\Lambda_1,
	\end{equation}
where $\Lambda_1 =\sum\limits_{j = 0}^{m - 1}( {f_j}(x){{\bar f}_{m + j}}(x) - {f_{m + j}}(x){{\bar f}_j}(x) )$ and $m=\frac{\ell}{2}$.
\end{theorem}
\begin{proof}
	Let $a(x)$ and $b(x)$ be arbitrary polynomials in $\mathbb{R}_{q,n}$. Then, any two codewords $\mathbf{c_{1}}$ and $\mathbf{c_{2}}$ of $\mathscr{C}$ can be denoted as $\mathbf{c_{1}}=(\left[a(x) f_{0}(x) g(x)\right],$ $\left[a(x) f_{1}(x) g(x)\right],$ $\ldots, \left[a(x) f_{\ell-1}(x) g(x)\right])$ 
	and $\mathbf{c_{2}}=(\left[b(x) f_{0}(x) g(x)\right],\left[b(x) f_{1}(x) g(x)\right],\ldots, 
	\left[b(x)f_{\ell-1}(x)g(x)\right])$, respectively.
	The symplectic inner product between $\mathbf{c_{1}}$ and $\mathbf{c_{2}}$ is 
		\begin{equation} 
			\begin{array}{rl}
		\left\langle \mathbf{c_{1}}, \mathbf{c_{2}}\right\rangle_{s}&=\mathbf{c_{1}} \cdot \Omega _{mn} \cdot \mathbf{c_{2}}^{T} \\
		&= \sum\limits_{j = 0}^{m - 1} {{{\left\langle {\left[ {a(x)g(x){f_j}(x)} \right],\left[ {b(x)g(x){f_{m + j}}(x)} \right]} \right\rangle }_e}}  \\
		&- \sum\limits_{j = 0}^{m - 1} {{{\left\langle {\left[ {a(x)g(x){f_{m + j}}(x)} \right],\left[ {b(x)g(x){f_j}(x)} \right]} \right\rangle }_e}}  \hfill \\
		&= \sum\limits_{j = 0}^{m - 1} {{{\left\langle {\left[ {a(x)g(x){f_j}(x){{\bar f}_{m + j}}(x)} \right],\left[ {b(x)g(x)} \right]} \right\rangle }_e}}  \\
		&- \sum\limits_{j = 0}^{m - 1} {{{\left\langle {\left[ {a(x)g(x){f_{m + j}}(x){{\bar f}_j}(x)} \right],\left[ {b(x)g(x)} \right]} \right\rangle }_e}}  \hfill \\
		&= \left\langle \left[ a(x) g (x) \Lambda_1\right],\left[b(x)g(x)\right] \right\rangle_e.
	\end{array}	
\end{equation}

	By Theorem \ref{cornerstone}, for all $a(x),b(x)\in\mathbb{R}_{q,n}$, the above equation is equal to zero 
	 if and only if there exists
$g^{\perp e}(x) \mid g (x) \Lambda_1$. 
Thus, we complete the proof.
\end{proof}


The $h$-generator QC codes can be considered a generalization of $1$-generator QC codes. The definition of $h$-generator QC codes is given here.
\begin{definition}\label{def-h-QC}
	Let  $1 \leq i \leq h$. 
	Let $g_i(x)$ be a polynomial in $\mathbb{R}_{q,n}$ that divides ${{x}^{n}}-1$, and let $f_{i,j}(x)$ be $\ell$ different polynomials in $\mathbb{R}_{q,n}$ for $0\le j\le \ell-1$, such that the greatest common divisor of $f_{i,0}(x),$ $f_{i,1}(x),\ldots,f_{i,{\ell-1}}(x)$ and $\frac{x^n-1}{g_i(x)}$ is equal to $1$.  
	If $\mathscr{C}$ is a QC code with $h$ generators given by:
	$(\left[g_{1}(x) f_{1,0}(x)\right],\left[g_{1}(x) f_{1,1}(x)\right],\ldots,\left[g_{1}(x) f_{1, \ell-1}(x)\right])$, 
	$(\left[g_{2}(x) f_{2,0}(x)\right],$ $\left[g_{2}(x) f_{2,1}(x)\right],\ldots,\left[g_{2}(x) f_{2, \ell-1}(x)\right]),\ldots$, 
	$(\left[g_{h}(x) f_{h, 0}(x)\right],$ $\left[g_{h}(x) f_{h, 1}(x)\right],\ldots,\left[g_{h}(x) f_{h, \ell-1}(x)\right]) \in \mathbb{R}_{q,n}^{\ell} $, then $\mathscr{C}$ is an $h$-generator QC code with index $\ell$, whose generator matrix is
	\begin{equation}\setlength{\arraycolsep}{1.2pt}
		G=\left(\begin{array}{cccc}
			G_{1,0} & G_{1,1} & \ldots & G_{1, \ell-1} \\
			G_{2,0} & G_{2,1} & \ldots & G_{2, \ell-1} \\
			\vdots & \vdots & \ddots & \vdots \\
			G_{h, 0} & G_{h, 1} & \ldots & G_{h, \ell-1}
		\end{array}\right),
	\end{equation}
	where $G_{i, j}$ is an $n \times n$ circulant matrix generated by $\left[g_{i}(x) f_{i, j}(x)\right]$.
\end{definition}

$h$-generator QC codes possess a more complex algebraic structure compared to $1$-generator QC codes. The former can be regarded as a vertical join of $h$ individual $1$-generator QC codes.  The following theorem gives the necessary and sufficient conditions under which the $h$-generator QC codes with index even to be symplectic self-orthogonal.

\begin{theorem}\label{theo-h-QC} 
	If $\mathscr{C}$ is an $h$-generator QC code with index even $\ell$ in Definition \ref{def-h-QC}, then the necessary and sufficient conditions for $\mathscr{C}$ to be symplectic self-orthogonal are that for all $r,s \in\{1,2, \ldots, h\}$, the following equation holds:
 	 \begin{equation}\label{h_sym_quasi-cyclic}
			g_{r}^{\perp_{e}}(x) \mid g_{s}(x)\Lambda_{r,s}, 
	\end{equation} where  $\Lambda_{r,s}=\sum\limits^{m-1}_{i=0} (f_{s, i}(x) \bar{f}_{r, m+i}(x)-f_{s, m+i}(x) \bar{f}_{r, i}(x))$ and $m=\frac{\ell}{2}$.
\end{theorem}
\begin{proof}  
Let $\mathscr{C}_i$ be $1$-generator QC codes generated by $(\left[g_{i}(x) f_{i,0}(x)\right], \left[g_{i}(x) f_{i,1}(x)\right], \ldots, \left[g_{i}(x) f_{i, \ell-1}(x)\right])$, where $i \in \{1,2, \ldots, h\}$.  
Then, in accordance with Definition \ref{def-h-QC}, we have $\mathscr{C} = \mathscr{C}_1 + \mathscr{C}_2 + \cdots + \mathscr{C}_h$.  
  
If $\mathscr{C}$ is a symplectic self-orthogonal code, then for all $\mathbf{c_1}, \mathbf{c_2} \in \mathscr{C}$, $\left\langle \mathbf{c_1}, \mathbf{c_2} \right\rangle_s = 0$. Thus, for all $r, s \in \{1,2, \ldots, h\}$, $\mathscr{C}_r$ and $\mathscr{C}_s$ are symplectic orthogonal to each other.  
  
Conversely, if for all $r, s \in \{1,2, \ldots, h\}$, $\mathscr{C}_r$ and $\mathscr{C}_s$ are symplectic orthogonal to each other, then for all $\mathbf{c_1}, \mathbf{c_2} \in \mathscr{C}$, $\left\langle \mathbf{c_1}, \mathbf{c_2} \right\rangle_s = 0$.   
  
Therefore, $\mathscr{C}$ is a symplectic self-orthogonal code if and only if $\mathscr{C}_r$ and $\mathscr{C}_s$ are symplectic orthogonal to each other for all $r, s \in \{1,2, \ldots, h\}$.  
  
Moreover, using a similar method as the proof of Theorem \ref{one_quasi-cyclic}, we have that $\mathscr{C}_r$ and $\mathscr{C}_s$ are symplectic orthogonal to each other if and only if  
\begin{equation}  
g_{r}^{\perp_{e}}(x) \mid g_{s}(x)\Lambda_{r,s}.  
\end{equation}  
Thus, we complete the proof.  
\end{proof}

As a consequence, for all QC codes with index even, one can directly deduce related necessary and sufficient conditions for symplectic self-orthogonality from Theorems \ref{one_quasi-cyclic} and \ref{theo-h-QC}.

\begin{example}  
	Set $q=2$ and $n=40$. Let $g(x)=x^5 + x^4 + x + 1$,   
	$f_0(x)=x^{34} + x^{33} + x^{32} + x^{30} + x^{29} + x^{28} + x^{26} + x^{24} + x^{23} + x^{22} + x^{19} + x^{16} + x^{15} + x^{14} + x^{12} + x^{10} + x^9 + x^8 + x^6 + x^5 + x^4$, and   
	$f_1(x)=x^{37} + x^{35} + x^{34} + x^{30} + x^{29} + x^{23} + x^{20} + x^{18} + x^{15} + x^9 + x^8 + x^4 + x^3 + x$.   
	It is straightforward to verify that $g(x)$, $f_0(x)$ and $f_1(x)$ satisfy the condition in Theorem \ref{one_quasi-cyclic}, namely, $g^{\perp_{e}}(x) \mid g(x)(f_{0}(x) \bar{f}_{1}(x)-f_{1}(x) \bar{f}_{0}(x))$, so we can obtain an $\left[80,35\right]_2^s$ symplectic self-orthogonal code $\mathscr{C}$.   
	Using Magma \cite{bosma1997magma}, we compute that the minimum symplectic dual distance of $\mathscr{C}$ is $10$.  
	  
	According to Theorem \ref{CRSS}, a binary QECC with parameters $\left[\left[40,5,10\right]\right]_2$ can be constructed.    
	It is noted that the best-known binary QECC with length $40$ and dimension $5$ in \cite{Grassltable} has parameters $\left[\left[40,5,9\right]\right]_2$. Thus, the minimum distance record for $\left[\left[40,5\right]\right]_2$ QECCs can be improved to $10$.  
\end{example}


%

\section{Bounds on the minimum symplectic distance of $1$-generator quasi-cyclic codes with index two and their symplectic dual codes}\label{sec4}
This section presents lower and upper bounds on the minimum symplectic distance of $1$-generator QC codes with index two and their symplectic dual codes. 
In particular, throughout this section, we consider the case where the index $\ell=2$ and $1$-generator QC code $\mathscr{C}$ has generator $(\left[f_0(x)g(x)\right],\left[f_1(x)g(x)\right])$, where $g(x)$ and $f_i(x)$ satisfy the conditions stated in Definition \ref{one_quasi-cyclic_def}. The parity check polynomial of $g(x)$ is denoted as $h(x)=\frac{x^n-1}{g(x)}$. 

\begin{definition}
	Let $f_a(x)$, $f_b(x)$ be polynomials in $\mathbb{R}_{q,n}$ and $t$ be a positive integer. Define $\mathrm{Circ}_t\left( \left[f_a(x)\right],\left[f_b(x)\right] \right)$ as the following matrix:  
	\begin{equation}
		G_t =
		\left(\begin{array}{c|c}
		  \pi(f_a(x))           &  \pi(f_b(x))        \\
			
			\pi(xf_a(x))        & \pi(xf_b(x))        \\
			
			\vdots           & \vdots           \\
			
			\pi(x^{t-1} f_a(x)) & \pi(x^{t-1} f_b(x))
		\end{array}\right)_{t\times 2n}.
	\end{equation}
\end{definition}

\begin{lemma}\label{QCmatrix_spread}
	With the above notation.
	Let $\mathscr{C}$ be a $1$-generator QC code with generator $(\left[f_0(x)g(x)\right],\left[f_1(x)g(x)\right])$, and let $t=n-\deg(g(x))-\deg(\gcd( h(x),f_1(x)))-\deg(\gcd( h(x),f_0(x)))$.
	Suppose that $\mathscr{C}_1$, $\mathscr{C}_2$ and $\mathscr{C}_3$ are codes generated by $\left[\frac{x^n-1}{\gcd( h(x),f_1(x))}\right]$, $\left[ \frac{x^n-1}{\gcd( h(x),f_0(x))}\right]$ and $\mathrm{Circ}_t\left( \left[g(x)f_0(x)\right],\left[g(x)f_1(x)\right] \right)$, where $h(x)=\frac{x^n-1}{g(x)}$. Then,  $\mathscr{C}=\mathscr{C}_1\oplus \mathscr{C}_2 + \mathscr{C}_3$.
\end{lemma}
\begin{proof}
Since the generator of $\mathscr{C}$ is $(\left[f_0(x)g(x)\right],\left[f_1(x)g(x)\right])$, any codeword of $\mathscr{C}$ has the form  $\mathbf{c}=(\mathbf{c}^{\prime},\mathbf{c}^{{\prime}{\prime}})=(\left[a(x)g(x)f_0(x)\right],$ $\left[a(x)g(x)f_1(x)\right])$, where $a(x)\in \mathbb{R}_{q,n}$. 
	This implies that only the following three types of nonzero codewords may appear in $\mathscr{C}$. We represent them with three sets:
	\begin{equation}
		\begin{matrix}
			S_1=\{ (\mathbf{c}^{\prime},\mathbf{c}^{{\prime}{\prime}})\mid \mathbf{c}^{\prime}\ne\mathbf{0}_n , \mathbf{c}^{{\prime}{\prime}} =\mathbf{0}_n \}, \\
			S_2=\{ (\mathbf{c}^{\prime},\mathbf{c}^{{\prime}{\prime}})\mid \mathbf{c}^{\prime}=\mathbf{0}_n , \mathbf{c}^{{\prime}{\prime}}\ne \mathbf{0}_n\}, \\
			S_3=\{ (\mathbf{c}^{\prime},\mathbf{c}^{{\prime}{\prime}})\mid \mathbf{c}^{\prime}\ne\mathbf{0}_n , \mathbf{c}^{{\prime}{\prime}}\ne \mathbf{0}_n\}.
		\end{matrix}
	\end{equation}

	If $\mathbf{c} \in S_1$, then $a(x)g(x)f_0(x) \not\equiv 0 \pmod{x^n-1}$ and $a(x)g(x)f_1(x)\equiv 0 \pmod{x^n-1}$. We have $\frac{x^n-1}{\gcd(x^n-1,g(x)f_1(x))}\mid a(x)$ and $a(x)$ can be written as $a(x)\equiv r(x)\frac{x^n-1}{\gcd(x^n-1,g(x)f_1(x))} \pmod{x^n-1}$, where $r(x)$ is a polynomial in $\mathbb{R}_{q,n}$ of degree less than $\deg(\gcd(x^n-1,g(x)f_1(x)))$. Since $\gcd(f_0(x),f_1(x),h(x))=1$ and $\frac{x^n-1}{\gcd(x^n-1,g(x)f_1(x))}g(x)=\frac{x^n-1}{\gcd(h(x),f_1(x))}$, we get $a(x)g(x)f_0(x) \equiv 0 \pmod{\frac{x^n-1}{\gcd(h(x),f_1(x))}}$.
 At this time, the cyclic code determined by $\left[a(x)g(x)f_0(x)\right]$ is indeed $\mathscr{C}_1$, which has dimension $\deg(\gcd( h(x),f_1(x)))$.
	
	Similarly, if $\mathbf{c} \in S_2$, then $a(x)g(x)f_0(x)\equiv 0 \pmod{x^n-1}$ and $a(x)g(x)f_1(x) \not\equiv 0 \pmod{x^n-1}$. 
 Therefore, when running through all $a(x)$ that divide $\frac{x^n-1}{\gcd(x^n-1,g(x)f_0(x))}$, the cyclic code that composed of all codewords $\left[a(x)g(x)f_0(x)\right]$ is $\mathscr{C}_2$, whose dimension is $\deg(\gcd( h(x),f_0(x)))$.
	Therefore, $\mathscr{C}_1\oplus \mathscr{C}_2 \subset \mathscr{C}$ and $\dim(\mathscr{C}_1\oplus \mathscr{C}_2)=\deg(\gcd( h(x),f_1(x)))+\deg(\gcd( h(x),f_0(x)))$.
	
	Given that $t=n-\deg(g(x))-\deg(\gcd( h(x),f_1(x)))-\deg(\gcd( h(x),f_0(x)))$, the nonzero codewords generated by $\mathrm{Circ}_t\left( \left[g(x)f_0(x)\right],\left[g(x)f_1(x)\right] \right)$ are all contained in $S_3$, and the dimension of $\mathscr{C}_3$ is $t$. Therefore, we have $\mathscr{C}_3\subset \mathscr{C}$ and $\dim(\mathscr{C}_1\oplus \mathscr{C}_2+\mathscr{C}_3)=n-\deg(g(x))$. 
	Therefore, the generator matrix of $\mathscr{C}$ can be written in the following form:
		\begin{equation}\label{QCtransform}
	G^{\prime}=
	\left(\begin{array}{c|c}
		\pi( g(x)f_0(x))                                                                  & \pi( g(x)f_1(x))                                                             \\
		
		\pi(x g(x)f_0(x))                                                                 & \pi(x g(x)f_1(x))                                                            \\
		
		\vdots                                                                           & \vdots                                                                      \\
		
		\pi(x^{t-1} g(x)f_0(x))         & \pi(x^{t-1}g(x)f_1(x))     \\
		\hline
		\mathbf{0}_n                                                                     & \pi( \frac{x^n-1}{\gcd( h(x),f_0(x))})                                  \\
		\mathbf{0}_n                                                                     & \pi(x \frac{x^n-1}{\gcd( h(x),f_0(x))} )                                \\
		\vdots                                                                           & \vdots                                                                      \\
		\mathbf{0}_n                                                                     & \pi(x^{t_a-1}   \frac{x^n-1}{\gcd( h(x),f_0(x))})\\		\hline
  		\pi( \frac{x^n-1}{\gcd( h(x),f_1(x))})                                 & \mathbf{0}_n                                                                \\
		
		\pi(x  \frac{x^n-1}{\gcd( h(x),f_1(x))})                               & \mathbf{0}_n                                                                \\
		
		\vdots                                                                           & \vdots                                                                      \\
		
		\pi(x^{t_b-1}  \frac{x^n-1}{\gcd( h(x),f_1(x))}) & \mathbf{0}_n                                                                \\

	\end{array}\right),\\
\end{equation}
	where $t_a=\deg(\gcd( h(x),f_1(x)))$ and $t_b=\deg(\gcd( h(x),f_0(x)))$. 
\end{proof}

 \begin{theorem}\label{SQCBound} 
	With the above notation,
	let $\mathscr{C}$ be a $1$-generator QC code with generator $(\left[f_0(x)g(x)\right],\left[f_1(x)g(x)\right])$.
 For $\alpha\in \mathbb{F}_q^*$, define
$\mathcal{I}^{(\alpha)}(\mathscr{C})=\gcd(f_0(x)+\alpha f_1(x),h(x))$. 
 Let $\mathcal{S} =\{\mathcal{I}^{(\alpha)}(\mathscr{C}) \mid \mathcal{I}^{(\alpha)}(\mathscr{C})\ne 1, \forall\alpha\in\mathbb{F}_q^*\}$.
	Suppose that $\mathscr{C}_0$, $\mathscr{C}_1$, $\mathscr{C}_2$, $\mathscr{C}_3$, $\mathscr{C}_4$ and $\mathscr{C}_5$ are cyclic codes determined by $[g(x)]$, $[g(x)f_0(x)]$, $[g(x)f_1(x)]$, $ \left[ \frac{x^n-1}{\gcd(h(x),f_0(x))}\right]$, $\left[\frac{x^n-1}{\gcd(h(x),f_1(x))}\right]$ and $\left[ g(x)\operatorname{lcm}(f_0(x),f_1(x))\right]$, respectively.
 Define $\mathcal{D}(\mathscr{C})$ as in Equation (\ref{The. 4. E.1}). \begin{figure*}
\begin{equation}\label{The. 4. E.1}
\mathcal{D}(\mathscr{C})=\max \left\{ \left\lceil\left(d_H(\mathscr{C}_1)+d_H(\mathscr{C}_2)+(q-|\mathcal{S} |-1)d_H(\mathscr{C}_0)+\sum\limits_{\mathcal{I}^{(\alpha)}(\mathscr{C})\in \mathcal{S} } d_H\left(\left[ g(x)\mathcal{I}^{(\alpha)}(\mathscr{C}) \right]\right)\right)/ q\right\rceil, d_H(\mathscr{C}_1),d_H(\mathscr{C}_2)\right\}.
\end{equation} 
\end{figure*}
Writing $d_{\mathrm{upper}}=	\min\left\{d_H(\mathscr{C}_3),d_H(\mathscr{C}_4)\right\}$ and   $d_{\mathrm{lower}}$ as in Equation (\ref{symplectic_bound}),
\begin{figure*}\begin{equation}\label{symplectic_bound}
	d_{\mathrm{lower}}=
	\left\{
	\begin{array}{ll}
		\min \left\{ d_H(\mathscr{C}_3),d_H(\mathscr{C}_4),\mathcal{D}(\mathscr{C}) \right\},& \mathcal{S} =\emptyset,                                                           \\
		\min\left\{ d_H(\mathscr{C}_3),d_H(\mathscr{C}_4), d_H(\mathscr{C}_5), \mathcal{D}(\mathscr{C})\right\},&q=2,\mathcal{S}\ne \emptyset,                     \\
	\min \left\{ d_H(\mathscr{C}_3),d_H(\mathscr{C}_4),\max\{d_H(\mathscr{C}_1),d_H(\mathscr{C}_2)\} \right\},&q>2, \mathcal{S}\ne \emptyset.                                        \\
	\end{array}
	\right.
\end{equation} \end{figure*}
then we have $d_{\mathrm{lower}}\le d_s(\mathscr{C})\le d_{\mathrm{upper}}$.
\end{theorem}
\begin{proof}
 Denote codeword of $\mathscr{C}$ by $\mathbf{c}=(\mathbf{c}^{\prime},\mathbf{c}^{\prime\prime})=(\left[a(x)g(x)f_0(x)\right],\left[a(x)g(x)f_1(x)\right])$, where $a(x)\in \mathbb{R}_{q,n}$. 
According to \cite{ling2010generalization}, we have
$q\mathbf{w}_{s}(\mathbf{c})=\mathbf{w}_{H}(\mathbf{c}^{\prime})+\mathbf{w}_{H}(\mathbf{c}^{\prime\prime})+\sum\limits_{\alpha\in \mathbb{F}_q^*} \mathbf{w}_{H}(\mathbf{c}^{\prime}+\alpha \mathbf{c}^{\prime\prime})$.

If $\mathbf{c}^{\prime}=\mathbf{0}_n$ or $\mathbf{c}^{\prime\prime}=\mathbf{0}_n$, then we have $a(x)g(x)f_0(x)\equiv  0 \pmod{x^n-1}$ or $a(x)g(x)f_1(x) \equiv  0 \pmod{x^n-1}$.
By Lemma \ref{QCmatrix_spread}, the nonzero half part of $\mathbf{c}$ is a codeword of $\mathscr{C}_3$ or $\mathscr{C}_4$. 
Therefore, it follows that 
 $d_s(\mathscr{C})\le \min \{d_H(\mathscr{C}_3), d_H(\mathscr{C}_4)\}$.

If $\mathbf{c}^{\prime}\ne \mathbf{0}_n$ and $\mathbf{c}^{\prime\prime}\ne \mathbf{0}_n$, then $a(x)g(x)f_0(x)\not\equiv 0 \pmod{x^n-1}$ and $a(x)g(x)f_1(x) \not\equiv 0 \pmod{x^n-1}$.
There are two cases to consider. 
Firstly, if for all $ \alpha$ in $\mathbb{F}_q^*$, $\gcd(f_0(x)+\alpha f_1(x),h(x))= 1$, i.e., $\mathcal{S} =\emptyset$,
then the following formula holds.
\begin{equation}\label{ESBound1}
	\begin{array}{rl}
		d_s(\mathscr{C})\ge&  (d_H(\left[ g(x)f_0(x)\right])+d_H(\left[g(x)f_1(x)\right])\\
		&+(q-1) d_H(\left[ g(x)\right] ))/q.
	\end{array}
\end{equation}

Secondly, if there is $\alpha$ in $\mathbb{F}_q^*$ such that $\mathcal{I}^{(\alpha)}(\mathscr{C})=\gcd(f_0(x)+\alpha f_1(x),h(x))\ne 1$, i.e., $\mathcal{S} \ne \emptyset$, then there exist  $a(x)$ of degree less than $n-\deg(g(x))$ such that $a(x)(f_0(x)+\alpha f_1(x)) g(x)\equiv 0 \pmod{x^n-1}$, i.e.,  $\mathrm{Supp}(\mathbf{c}^{\prime})=\mathrm{Supp}(\mathbf{c}^{\prime\prime})$. 
For $q=2$, if $\mathrm{Supp}(\mathbf{c}^{\prime})=\mathrm{Supp}(\mathbf{c}^{\prime\prime})$, then we have $\mathbf{c}^{\prime}=\mathbf{c}^{\prime\prime}$ $\in \mathscr{C}_1 \cap \mathscr{C}_2$. Therefore, 
$d_s(\mathscr{C}) \ge d_H( \left[ g(x) \operatorname{lcm}(f_0(x),f_1(x)) \right] )$.  
For $q> 2$, we can only determine $d_s(\mathscr{C})\ge \max\{d_H(\mathscr{C}_1),d_H(\mathscr{C}_2)\}$.

If $\mathrm{Supp}(\mathbf{c}^{\prime})\ne \mathrm{Supp}(\mathbf{c}^{\prime\prime})$, then we have Equation (\ref{ESBound2}). 
\begin{figure*}\begin{equation}\label{ESBound2}
d_s(\mathscr{C})\ge \left(d_H(\mathscr{C}_1)+d_H(\mathscr{C}_2)+(q-|\mathcal{S} |-1)d_H(\mathscr{C}_0)+\sum\limits_{\mathcal{I}^{(\alpha)}(\mathscr{C})\in \mathcal{S} } d_H\left(\left[ g(x)\mathcal{I}^{(\alpha)}(\mathscr{C}) \right]\right)\right)/ q.
\end{equation}\end{figure*}

Additionally, $d_s(\mathscr{C})$ must also satisfy the condition that it is greater than or equal to $\max\left\{ d_H(\mathscr{C}_1),d_H(\mathscr{C}_2) \right\}$ when Equations (\ref{ESBound1}) and (\ref{ESBound2}) are satisfied. 
Therefore, we have $d_s(\mathscr{C})\ge\mathcal{D}(\mathscr{C})$.

From the combination of the above conditions one can obtain upper and lower bounds for the minimum symplectic distance of $\mathscr{C}$. This concludes the proof.
\end{proof}


\begin{remark}\label{improve_bound}
	Since $1$-generator QC codes with index two over $\mathbb{F}_q$ are equivalent with additive cyclic codes over $\mathbb{F}_{q^2}$ \cite{GuanAdditive2023}, Theorem \ref{SQCBound} also can be viewed as bounds on the minimum Hamming distance on additive cyclic codes.
	In addition, Corollary 3.9 in \cite{dastbasteh2023polynomial} can be viewed as a subcase ($q>2,\mathcal{S} \ne\emptyset$) of Theorem \ref{SQCBound},  
	and Theorem 1 in \cite{GuanAdditive2023} is a simplified form of $\mathcal{S} =\emptyset$ in Theorem \ref{SQCBound}.
	Thus, Theorem \ref{SQCBound} improves and enhances the results in \cite{dastbasteh2023polynomial} and \cite{GuanAdditive2023}. 
\end{remark}

\begin{example}\label{good_Ex1}
	Take $q=2$ and $n=21$. Set $g(x)=x^6 + x^5 + x^4 + x^2 + 1$, $f_0(x)=x^{11} + x^6 + x^5 + x^2 + x + 1$ and $f_1(x)=x^{12} + x^7 + x^6 + x^5 + x^2 + x + 1$, then $h(x)=\frac{x^{21}-1}{g(x)}=x^{15} + x^{14} + x^{12} + x^9 + x^8 + x^5 + x^2 + 1$. Since $\gcd(f_0(x),f_1(x),h(x))=1$ and $\deg(h(x))=15$, $([g(x)f_0(x)],[g(x)f_1(x)])$ generates a $\left[42,15\right]_2$ QC code $\mathscr{C}$.

 Given that $\mathcal{S}=\{ \mathcal{I}^{(1)}(\mathscr{C})=\gcd(f_0(x)+f_1(x),h(x))=x^5 + x^4 + 1\}$, the lower bound on the minimum symplectic distance of $\mathscr{C}$ is the second condition in Equation (\ref{symplectic_bound}), i.e., 
 \begin{equation}
	\begin{array}{rl}
	 \min\left\{ d_H(\mathscr{C}_3),d_H(\mathscr{C}_4), d_H(\mathscr{C}_5), \mathcal{D}(\mathscr{C}) \right\}  \le d_s(\mathscr{C}) &\\
	 \le\min\left\{d_H(\mathscr{C}_3),d_H(\mathscr{C}_4)\right\}&
	\end{array}
\end{equation}
  where $\mathcal{D}(\mathscr{C})=\max \left\{ \left\lceil\left(d_H(\mathscr{C}_1)+d_H(\mathscr{C}_2)+d_H(\left[g(x)\mathcal{I}^{(1)}(\mathscr{C})\right])\right)/ 2\right\rceil\right.$, 
  $\left.d_H(\mathscr{C}_1),d_H(\mathscr{C}_2)\right\}$. 
Here we list the generator polynomials and parameters of the cyclic codes $\mathscr{C}_i$:
 \begin{itemize}
      \item  $\mathscr{C}_0$: Generator polynomial $\left[g(x)\right]$. Parameters $\left[21,15,3\right]_2$.
      \item  $\mathscr{C}_1$: Generator polynomial $\left[g(x)f_0(x)\right]$. Parameters $\left[21,14,4\right]_2$.
      \item  $\mathscr{C}_2$: Generator polynomial $\left[g(x)f_1(x)\right]$. Parameters $\left[21,6,7\right]_2$.
     \item $\mathscr{C}_3$: Generator polynomial $\left[\frac{x^{21}-1}{\gcd(h(x),f_0(x))}\right]=\left[\frac{x^{21}-1}{x+1}\right]$. Parameters $\left[21,1,21\right]_2$.
      \item $\mathscr{C}_4$: Generator polynomial $\left[\frac{x^{21}-1}{\gcd(h(x),f_1(x))}\right]=\left[x^{12} + x^{11} + x^9 + x^7 + x^3 + x^2 + x + 1\right]$. Parameters $\left[21,9,8\right]_2$.
     \item $\mathscr{C}_5$: Generator polynomial $\left[g(x)\operatorname{lcm}(f_0(x),f_1(x))\right]=\left[x^{16} + x^{15} + x^{14} + x^{13} + x^{12} + x^{10} + x^8 + x^5 + x^4 + 1\right]$. Parameters $\left[21,5,10\right]_2$.
     \item   Generator polynomial $\left[g(x)\mathcal{I}^{(1)}(\mathscr{C})\right]=\left[x^{11} + x^8 + x^7 + x^2 + 1\right]$. Parameters $\left[21,10,5\right]_2$. 
 \end{itemize}

  

We can directly derive that the lower and upper bounds on the minimum symplectic distance of $\mathscr{C}$ are both $8$ based on the parameters of the above cyclic codes. 
 By utilising Magma \cite{bosma1997magma}, it can be determined that the specific parameters of $\mathscr{C}$ are indeed $\left[42,15,8\right]_2^s$. 
 In contrast, Corollary 3.9 in \cite{dastbasteh2023polynomial} provides a lower bound of $4$ on the minimum symplectic distance of $\mathscr{C}$, whereas Theorem 1 in \cite{GuanAdditive2023} does not yield a lower bound for this code. Consequently, our bound is more sharper and applicable.
\end{example}

The following example shows that Theorem \ref{SQCBound} performs better even though the constraints in Theorem 1 of \cite{GuanAdditive2023} are satisfied.

\begin{example}\label{good_Ex_for additive}
	Let $q=2$ and $n=31$. Set $g(x)=x^5 + x^2 + 1$, $f_0(x)=x^{25} + x^{24} + x^{23} + x^{22} + x^{21} + x^{20} + x^{19} + x^{16} + x^{13} + x^{11} + x^9 +x^8 + x^7 + x^6 + 1$ and $f_1(x)=x^{21} + x^{20} + x^{15} + x^{13} + x^8 + x^5 + x^4 + x^3$, then $h(x)=\frac{x^{31}-1}{g(x)}=x^{26} + x^{23} + x^{21} + x^{20} + x^{17} + x^{16} + x^{15} + x^{14} + x^{13} + x^9 + x^8 + x^6 +
	x^5 + x^4 + x^2 + 1$. Since $\gcd(f_0(x),f_1(x),h(x))=1$ and $\deg(h(x))=26$, $([g(x)f_0(x)],[g(x)f_1(x)])$ produces a $\left[62,26\right]_2$ QC code $\mathscr{C}$.

 Given that $\mathcal{S}=\emptyset$, the lower bound on the minimum symplectic distance of $\mathscr{C}$ is the first condition in Equation (\ref{symplectic_bound}), i.e., 
  \begin{equation}
 	\begin{array}{rl}
 		\min\left\{ d_H(\mathscr{C}_3),d_H(\mathscr{C}_4), \mathcal{D}(\mathscr{C})\right\} \le d_s(\mathscr{C}) &\\
 		\le \min\left\{d_H(\mathscr{C}_3),d_H(\mathscr{C}_4)\right\}&
 	\end{array}
 \end{equation}
 where $\mathcal{D}(\mathscr{C})=\max \left\{ \left\lceil \left(d_H(\mathscr{C}_0)+d_H(\mathscr{C}_1)+d_H(\mathscr{C}_2)\right)/2\right\rceil \right.,$ $\left.d_H(\mathscr{C}_1),d_H(\mathscr{C}_2)\right\}$.
Here we list the generator polynomials and parameters of the cyclic codes $\mathscr{C}_i$:
 \begin{itemize}
      \item  $\mathscr{C}_0$: Generator polynomial $\left[g(x)\right]$. Parameters $\left[31,26,3\right]_2$.
      \item  $\mathscr{C}_1$: Generator polynomial $\left[g(x)f_0(x)\right]$. Parameters $\left[31,16,7\right]_2$.
      \item  $\mathscr{C}_2$: Generator polynomial $\left[g(x)f_1(x)\right]$. Parameters $\left[31,25,4\right]_2$.
     \item $\mathscr{C}_3$: Generator polynomial $\left[\frac{x^{31}-1}{\gcd(h(x),f_0(x))}\right]=\left[x^{21} +\right. $ $ \left. x^{18} + x^{17} + x^{15} +x^{13} + x^{10} + x^5 + x^4 + x^3 + x^2 + x \right. $ $ \left.+ 1\right]$. Parameters $\left[31,10,12\right]_2$.
      \item $\mathscr{C}_4$: Generator polynomial $\left[\frac{x^{31}-1}{\gcd(h(x),f_1(x))}\right]=\left[\frac{x^{31}-1}{x+1}\right]$. Parameters $\left[31,1,31\right]_2$.
 \end{itemize}

	
	We can directly derive lower and upper bounds on the minimum symplectic distance of $\mathscr{C}$ are $7$ and $12$ based on the parameters of above cyclic codes. By using Magma \cite{bosma1997magma}, we can calculate that the specific parameters of $\mathscr{C}$ are $\left[62,26,11\right]_2^s$. 
 In contrast, Theorem 1 in \cite{GuanAdditive2023} provides a lower bound of $5$ on the minimum symplectic distance of $\mathscr{C}$. Hence, our bound is more tighter. 
\end{example}

\begin{theorem}\label{Dual_Structure} 
	With the above notation,
	let $\mathscr{C}$ be a $1$-generator QC code with generator $(\left[f_0(x)g(x)\right],\left[f_1(x)g(x)\right])$.
	 If $\gcd(f_0(x),g(x))=1$,		
then the symplectic dual code $\mathscr{C}^{\perp_s}$ of $\mathscr{C}$ is a 2-generator QC code, whose generator matrix is
\begin{equation}
	G^{\perp_{s}}=
	\begin{pmatrix}
		F_{0}	& F_{1}\\
		\mathbf{0}&G_{g^{\perp_e}}
	\end{pmatrix},
\end{equation}
where $(F_{0},F_1)$ and $(\mathbf{0},G_{g^{\perp_e}})$ are the generator matrices of $1$-generator QC codes determined by $(\left[ \bar{f}_{0}(x)\right], \left[ \bar{f}_{1}(x)\right])$ and $(\left[ 0\right], \left[ g^{\perp_e}(x)\right])$, respectively. 
\end{theorem}
\begin{proof}
	 Let $\mathscr{C}^{\circ}$ denote the code generated by $G^{\perp_{s}}$.  
 Then any codeword of $\mathscr{C}$ and $\mathscr{C}^{\circ}$ has the form $\mathbf{c_{1}}=(\left[a(x) f_{0}(x) g(x)\right],$ $\left[a(x) f_{1}(x) g(x)\right])$ and $\mathbf{c_{2}}=(\left[b(x) \bar{f}_{0}(x)\right],\left[b(x) \bar{f}_{1}(x)+c(x)g^{\perp_e}(x)\right])$, where $a(x),b(x),c(x)\in\mathbb{R}_{q,n}$. 
	Using a similar approach of Theorem \ref{one_quasi-cyclic}, it is easy to prove that $\left\langle \mathbf{c_{1}}, \mathbf{c_{2}}\right\rangle_{s}=\mathbf{c_{1}} \cdot \Omega _{mn} \cdot \mathbf{c_{2}}^{T}=0$.
	Therefore, $\mathscr{C}$ and $\mathscr{C}^{\circ}$ are mutually symplectic orthogonal.


 		By Lemma \ref{QCmatrix_spread}, the generator matrix of $\mathscr{C}^{\circ}$ also can be transformed into 
\begin{equation}
	G_1^{\perp_{s}}=
	\begin{pmatrix}
		F_{0}^{*}	& F_{1}^{*}\\
		A	& \mathbf{0}\\
		\mathbf{0}	& B\\
		\mathbf{0}&G_{g^{\perp_e}}
	\end{pmatrix},
\end{equation}
 		where $(F_{0}^{*},F_1^{*})$ denotes the first $(n-\deg(\gcd(x^n-1,\bar{f}_1(x)))-\deg(\gcd(x^n-1,\bar{f}_{0}(x))))$-row of $(F_{0},F_1)$, $A$ and $B$ are the generator matrices of cyclic codes generated by $\left[  \frac{x^n-1}{\gcd(x^n-1,\bar{f}_{1}(x))}\right]$ and $\left[ \frac{x^n-1}{\gcd(x^n-1,\bar{f}_{0}(x))}\right]$, respectively.
 		
 		According to Definition \ref{one_quasi-cyclic_def}, $\gcd(\bar{f}_{0}(x),\bar{f}_{1}(x),g^{\perp_e}(x))=1$. Given that $\gcd(f_0(x),g(x))=1$, we have $\gcd(\bar{f}_0(x),\tilde{g}(x) )=1$. As $\tilde{g}(x)g^{\perp_e}(x)=x^n-1$,  $\gcd(\bar{f}_{0}(x),\bar{f}_{1}(x),x^n-1)=1$ and $\operatorname{lcm}\left(g^{\perp_e}(x),\frac{x^n-1}{\gcd(x^n-1,\bar{f}_{0}(x))}\right) \equiv 0 \pmod{x^n-1}$ hold. This implies that, the rank of $(F_{0}, F_{1})$ is $n$ and the rows of $B$ and $G_{g^{\perp_e}}$ are linearly independent. Therefore, the rank of $G_1^{\perp_{s}}$ is $n+\deg(g(x))$. Thus, $\mathscr{C}$ and $\mathscr{C}^{\circ}$ are symplectic dual codes.
\end{proof}
%
%


\begin{theorem}\label{SDQCBound} With the above notation,
	let $\mathscr{C}$ be a $1$-generator QC code with generator $(\left[f_0(x)g(x)\right],\left[f_1(x)g(x)\right])$ satisfying
 	$\gcd(f_0(x),g(x))=1$. 
  For $\alpha\in \mathbb{F}_q^*$, define $\mathcal{I}^{(\alpha)}_{\perp}(\mathscr{C})=\gcd(\bar{f}_1(x)+\alpha  \bar{f}_0(x),g^{\perp_{e}}(x))$. Let $\mathcal{S}_{\perp} =\{\mathcal{I}^{(\alpha)}_{\perp}(\mathscr{C}) \mid \mathcal{I}^{(\alpha)}_{\perp}(\mathscr{C})\ne 1, \forall\alpha\in\mathbb{F}_q^*\}$.
  
	Suppose that $\mathscr{C}_0$, $\mathscr{C}_1$, $\mathscr{C}_2$, $\mathscr{C}_3$, $\mathscr{C}_4$, $\mathscr{C}_5$, $\mathscr{C}_6$ and $\mathscr{C}_7$ are cyclic codes determined by 
 $\left[g^{\perp_{e}}(x)\right]$,
	$\left[\bar{f}_0(x)\right]$, $\left[\gcd(\bar{f}_1(x),g^{\perp_{e}}(x))\right]$,
	$\left[\frac{x^n-1}{\gcd( x^n-1,\bar{f}_1(x))}\right]$, 	$\left[\frac{x^n-1}{\gcd( x^n-1,\bar{f}_0(x))}\right]$, $\left[\gcd\left( \frac{x^n-1}{\gcd( x^n-1,\bar{f}_0(x))},g^{\perp_{e}}(x) \right)\right]$, $\left[ \operatorname{lcm}(\bar{f}_0(x),\gcd(\bar{f}_1(x),g^{\perp_{e}}(x))\right]$ and $\left[\bar{f}_0(x)\frac{g^{\perp_{e}}(x)}{\gcd(\bar{f}_1(x),g^{\perp_{e}}(x))}\right]$, respectively.

Define $\mathcal{D}_{\perp}(\mathscr{C})$ as in Equation (\ref{The.6.Eq.1}).
\begin{figure*}
	\begin{equation}\label{The.6.Eq.1}
		\mathcal{D}_{\perp}(\mathscr{C})=\max \left\{ \left\lceil\left(d_H(\mathscr{C}_1)+d_H(\mathscr{C}_2)+(q-|\mathcal{S}_{\perp} |-1)+\sum\limits_{\mathcal{I}^{(\alpha)}_{\perp}(\mathscr{C})\in \mathcal{S}_{\perp} } d_H\left(\left[ \mathcal{I}^{(\alpha)}_{\perp}(\mathscr{C}) \right]\right)\right)/ q\right\rceil, d_H(\mathscr{C}_1),d_H(\mathscr{C}_2)\right\}.
	\end{equation}
\end{figure*}
Writing $d_{\mathrm{upper}}=\min\left\{d_H(\mathscr{C}_0),d_H(\mathscr{C}_3),d_H(\mathscr{C}_4)\right\}$ and  $d_{\mathrm{lower}}$ as in Equation (\ref{Dual_lower_Bound}),
\begin{figure*}\begin{equation}\label{Dual_lower_Bound}
	d_{\mathrm{lower}}=
	\left\{
	\begin{array}{ll}
		\min \left\{ d_H(\mathscr{C}_3),d_H(\mathscr{C}_5),\mathcal{D}_{\perp}(\mathscr{C}) \right\},& \mathcal{S}_{\perp} =\emptyset, g(x)\mid \bar{f}_1(x),                                                           \\
		\min \left\{ d_H(\mathscr{C}_3),d_H(\mathscr{C}_5),d_H(\mathscr{C}_7),\mathcal{D}_{\perp}(\mathscr{C}) \right\},& \mathcal{S}_{\perp} =\emptyset, g(x)\nmid \bar{f}_1(x),                                       \\
		\min\left\{ d_H(\mathscr{C}_3),d_H(\mathscr{C}_5), d_H(\mathscr{C}_6), \mathcal{D}_{\perp}(\mathscr{C})\right\},&q=2,\mathcal{S}_{\perp} \ne\emptyset, g(x)\mid \bar{f}_1(x),                                   \\
		\min\left\{ d_H(\mathscr{C}_3),d_H(\mathscr{C}_5), d_H(\mathscr{C}_6),d_H(\mathscr{C}_7), \mathcal{D}_{\perp}(\mathscr{C})\right\},& q=2,\mathcal{S}_{\perp} \ne\emptyset,g(x)\nmid \bar{f}_1(x),\\
		\min \left\{ d_H(\mathscr{C}_3),d_H(\mathscr{C}_5),\max\{d_H(\mathscr{C}_1),d_H(\mathscr{C}_2)\} \right\},&q>2,\mathcal{S}_{\perp}\ne \emptyset,g(x)\mid \bar{f}_1(x),                     \\
		\min \left\{ d_H(\mathscr{C}_3),d_H(\mathscr{C}_5),d_H(\mathscr{C}_7),\max\{d_H(\mathscr{C}_1),d_H(\mathscr{C}_2)\} \right\},&q>2,\mathcal{S}_{\perp}\ne \emptyset,g(x)\nmid \bar{f}_1(x). \\
	\end{array}
	\right.
\end{equation} \end{figure*}
then we have $d_{\mathrm{lower}}\le d_s(\mathscr{C}^{\perp_s}) \le d_{\mathrm{upper}}$.
\end{theorem}
\begin{proof}
Let $t=(n-\deg(\gcd(x^n-1,\bar{f}_1(x)))-\deg(\gcd(x^n-1,\bar{f}_{0}(x))))$ and $\mathscr{C}_8$ be the cyclic code generated by $\left[ \bar{f}_1(x) \right]$. 
	Since $\gcd(f_0(x),g(x))=1$,
	by Theorem \ref{Dual_Structure}, the symplectic dual of $\mathscr{C}$ is a 2-generator QC code with generator matrix
	\begin{equation}
		G^{\perp_{s}}
  =
  \begin{pmatrix}
		F_{0}^{*}	& F_{1}^{*}\\
		A	& \mathbf{0}\\
		\mathbf{0}	& B\\
		\mathbf{0}&G_{g^{\perp_e}}
	\end{pmatrix},
	\end{equation}
where
$(F_{0}^{*},F_1^{*})=\mathrm{Circ}_t\left( \left[\bar{f}_0(x)\right],\left[\bar{f}_1(x)\right] \right)$, $A$ and $B$ are the generator matrices of cyclic codes generated by $\left[  \frac{x^n-1}{\gcd(x^n-1,\bar{f}_{1}(x))}\right]$ and $\left[ \frac{x^n-1}{\gcd(x^n-1,\bar{f}_{0}(x))}\right]$, and $G_{g^{\perp_e}}$ is the generator matrix of cyclic code generated by $[g^{\perp_e}(x)]$, respectively.  
Accordingly, any codeword of $\mathscr{C}^{\perp_{s}}$ can be expressed as $\mathbf{c}=(\left[a(x)\bar{f}_0(x)\right],\left[a(x)\bar{f}_1(x)+b(x)g^{\perp_{e}}(x)\right])$, where $a(x),b(x)\in\mathbb{R}_{q,n}$. 
	By \cite{ling2010generalization}, we have
     $q\mathbf{w}_{s}(\mathbf{c})=\mathbf{w}_{H}(\left[a(x)\bar{f}_0(x)\right])+\mathbf{w}_{H}(\left[a(x)\bar{f}_1(x)+b(x)g^{\perp_{e}}(x)\right])+\sum\limits_{\alpha  \in \mathbb{F}_q^*} \mathbf{w}_{H}(\left[a(x)(\bar{f}_1(x)+\alpha   \bar{f}_0(x))+b(x)g^{\perp_{e}}(x)\right]).$

	Similar to the proof of Theorem \ref{SQCBound}, if $a(x)\bar{f}_1(x)\equiv  0 \pmod{x^n-1}$ or $a(x)\bar{f}_0(x) \equiv  0 \pmod{x^n-1}$, then $\mathbf{c}=(\left[a(x)\bar{f}_0(x)\right],\left[b(x)g^{\perp_{e}}(x)\right])$ or $(\mathbf{0}_n,\left[a(x)\bar{f}_1(x)+b(x)g^{\perp_{e}}(x)\right])$.
	At this point, we have $\min \left\{ d_H(\mathscr{C}_3),d_H(\mathscr{C}_5) \right\}  \le \mathbf{w}_{s}(\mathbf{c}) \le 
	\min\left\{d_H(\mathscr{C}_0),d_H(\mathscr{C}_3),d_H(\mathscr{C}_4)\right\}$.

  
 If $a(x)\bar{f}_0(x)\not\equiv 0 \pmod{x^n-1}$ and $a(x)\bar{f}_1(x) \not\equiv 0 \pmod{x^n-1}$, there are the following scenarios to discuss. 

	Firstly,  if $g(x)\mid \bar{f}_1(x)$, then $\operatorname{lcm}(\bar{f}_1(x),g^{\perp_e}(x))\equiv 0 \pmod{x^n-1}$. 
 It implies that $\mathscr{C}_0\cap \mathscr{C}_8=\{ \mathbf{0}_{n}\}$. 
	 If $g(x)\nmid \bar{f}_1(x)$, then the two cyclic codes $\mathscr{C}_0$ and $\mathscr{C}_8$ have nontrivial intersection.
 Therefore, codewords of the form $(\left[a(x)\bar{f}_0(x)\right],\mathbf{0}_n)$ may also exist in $\mathscr{C}^{\perp_s}$. 
 In this case, there exists $\frac{g^{\perp_{e}}(x)}{\gcd(\bar{f}_1(x),g^{\perp_{e}}(x))} \mid a(x)$. Thus, when $g(x)\nmid \bar{f}_1(x)$, we can only deduce $d_s(\mathscr{C}^{\perp_s})\ge d_H(\mathscr{C}_7)$. 
	
	Secondly, if for all $ \alpha  $ in $\mathbb{F}_q^*$, $\gcd(\bar{f}_1(x)+\alpha   \bar{f}_0(x),g^{\perp_{e}}(x))= 1$, i.e., $\mathcal{S}_{\perp} =\emptyset$, then 
	\begin{equation}\label{EDSBound1}
		d_s(\mathscr{C}^{\perp_s})\ge \left(d_H(\mathscr{C}_1)+d_H(\mathscr{C}_2)+(q-1) \right)/ q.
	\end{equation}
	
	Thirdly, if there exists $\alpha  $ in $\mathbb{F}_q^*$ such that  $\gcd(\bar{f}_1(x)+\alpha   \bar{f}_0(x),g^{\perp_{e}}(x))\ne 1$, i.e., $\mathcal{S}_{\perp} \ne\emptyset$, then there may exist $\mathrm{Supp}(\left[a(x)\bar{f}_0(x)\right]) =\mathrm{Supp}(\left[a(x)\bar{f}_1(x)+b(x)g^{\perp_{e}}(x)\right])$.
	Thus, for $q>2$, we can only deduce $d_s(\mathscr{C}^{\perp_s})\ge \max\left\{ d_H(\mathscr{C}_1),d_H(\mathscr{C}_2) \right\}$.
	But for $q=2$, we have
	$d_s(\mathscr{C}^{\perp_s})\ge d_H( \left[ \operatorname{lcm}(\bar{f}_0(x),\gcd(\bar{f}_1(x),g^{\perp_{e}}(x))) \right])$. 
	
	For the case of $\mathrm{Supp}(\left[a(x)\bar{f}_0(x)\right])\ne$ 
	$ \mathrm{Supp}(\left[a(x)\bar{f}_1(x)+\right. $ $\left.b(x)g^{\perp_{e}}(x)\right])$,
	we have Equation (\ref{EDSBound2}) holds.

	 \begin{figure*}
	\begin{equation}\label{EDSBound2}
		d_s(\mathscr{C}^{\perp_s})\ge \left(d_H(\mathscr{C}_1)+d_H(\mathscr{C}_2)+(q-|\mathcal{S}_{\perp} |-1)+\sum\limits_{\mathcal{I}^{(\alpha)}_{\perp}(\mathscr{C})\in \mathcal{S}_{\perp} } d_H\left(\left[ \mathcal{I}^{(\alpha)}_{\perp}(\mathscr{C}) \right]\right)\right)/ q.
	\end{equation} \end{figure*} 
	
	Moreover, at this time, $d_s(\mathscr{C}^{\perp_s})$ must be greater than or equal to $\max\left\{ d_H(\mathscr{C}_1),d_H(\mathscr{C}_2) \right\}$ when Equations (\ref{EDSBound1}) and (\ref{EDSBound2}) are satisfied.
 Therefore, we have $d_s(\mathscr{C})\ge\mathcal{D}_{\perp}(\mathscr{C})$.
 
	From the combination of the above conditions, one can obtain upper and lower bounds for the minimum symplectic distance of $\mathscr{C}^{\perp_s}$. This concludes the proof.
\end{proof}

\begin{remark}\label{improve_bound_Galindo}
Theorem \ref{SDQCBound} estimates the range of the minimum symplectic distance for $2$-generator QC codes $\mathscr{C}^{\perp_s}$ in Theorem \ref{Dual_Structure}. 
It should be noted that Galindo et al. also studied the lower bound on minimum symplectic distance for a class of  $2$-generator QC codes with generators $(\left[g_1(x)\right],\left[f(x)g_1(x)\right])$ and $([0],\left[g_2(x)\right])$, where $g_i(x) \mid (x^n-1)$ and $f(x)\in \mathbb{R}_{q,n}$ (please see Proposition 4 in \cite{galindo2018quasi}). 
It is easy to verify that if we constraint $g_1(x)=1$, then Galindo's codes are a subclass of $\mathscr{C}^{\perp_s}$ in Theorem \ref{Dual_Structure}.   Therefore,
Theorem \ref{SDQCBound} can be seen as a valid  generalization of Proposition 4 in \cite{galindo2018quasi}. 
\end{remark}

\begin{example}\label{good_Ex2}
	Let $q=2$ and $n=15$. Set $g(x)=x^4 + x + 1$, $f_0(x)=x^{13} + x^{12} + x^{11} + x^8 + x^7 + x^4 + x^3 + x^2 + 1$ and $f_1(x)=x^{13} + x^9 + x^8 + x^7 + x^6 + x^2 + 1$, then $g^{\perp_e}(x)=x^{11} + x^{10} + x^9 + x^8 + x^6 + x^4 + x^3 + 1$. It is easy to check that $\gcd(f_0(x),f_1(x))=1$, $\gcd(\bar{f}_0(x),g(x))=1$ and 
 $g^{\perp_e}(x) \mid g(x)(f_{0}(x) \bar{f}_{1}(x)-f_{1}(x) \bar{f}_{0}(x))$, so $(\left[g(x)f_0(x)\right],\left[g(x)f_1(x)\right])$ generates a symplectic self-orthogonal $\left[30,11\right]_2$ QC code $\mathscr{C}$.

Since $\mathcal{S}_{\perp} =\{\mathcal{I}^{(1)}_{\perp}(\mathscr{C})=\gcd(\bar{f}_1(x)+ \bar{f}_0(x),g^{\perp_{e}}(x))= x + 1\}$ and $\operatorname{lcm}(\bar{f}_1(x),g^{\perp_e}(x))\equiv 0 \pmod{x^n-1}$, we get

 	\begin{equation} \label{improve_boundD}
 			\begin{array}{rl}
 	\min\left\{ d_H(\mathscr{C}_3),d_H(\mathscr{C}_5), d_H(\mathscr{C}_6), \mathcal{D}_{\perp}(\mathscr{C})\right\}  \le d_s(\mathscr{C}^{\perp_s})\\
 	 \le 
 	\min\left\{d_H(\mathscr{C}_0),d_H(\mathscr{C}_3),d_H(\mathscr{C}_4)\right\},
 \end{array}
 \end{equation}
where 
 $
 	\mathcal{D}_{\perp}(\mathscr{C})=\max \left\{ \left\lceil\left(d_H(\mathscr{C}_1)+d_H(\mathscr{C}_2)+ d_H\left(\left[ \mathcal{I}^{(1)}_{\perp}(\mathscr{C}) \right]\right)\right)/ 2\right\rceil\right.
 	,$ $\left.  d_H(\mathscr{C}_1),d_H(\mathscr{C}_2)\right\}.
$
Here, we list the generator polynomials and parameters of the cyclic codes involved.
 \begin{itemize}
      \item  $\mathscr{C}_0$: Generator polynomial $\left[g^{\perp_{e}}(x)\right]$. 
      Parameters $\left[15,4,8\right]_2$.
      \item  $\mathscr{C}_1$: Generator polynomial $\left[\bar{f}_0(x)\right]=\left[x^6 + x^4 + x^3 + x^2 + 1\right]$. Parameters $\left[15,9,4\right]_2$.
      \item  $\mathscr{C}_2$: Generator polynomial $\left[\gcd(\bar{f}_1(x),g^{\perp_{e}}(x))\right]=\left[x^4 + x + 1\right]$. Parameters $\left[15,11,3\right]_2$.
     \item $\mathscr{C}_3$: Generator polynomial $\left[\frac{x^{15}-1}{\gcd( x^{15}-1,\bar{f}_1(x))}\right]=\left[x^7 + x^6 + x^5 + x^2 + x + 1\right]$. Parameters $\left[15,8,4\right]_2$.
      \item $\mathscr{C}_4$: Generator polynomial $\left[\frac{x^{15}-1}{\gcd( x^{15}-1,\bar{f}_0(x))}\right]=\left[x^9 + x^7 + x^6 + x^3 + x^2 + 1\right]$. Parameters $\left[15,6,6\right]_2$.
      \item  $\mathscr{C}_5$: Generator polynomial $\left[\gcd\left( \frac{x^{15}-1}{\gcd( x^{15}-1,\bar{f}_0(x))},g^{\perp_{e}}(x) \right)\right]$ $=\left[x^5 + x^4 + x^2 + 1\right]$. Parameters $\left[15,10,4\right]_2$.
      \item  $\mathscr{C}_6$: Generator polynomial $\left[\operatorname{lcm}(\bar{f}_0(x),\gcd(\bar{f}_1(x),g^{\perp_{e}}(x))\right]=\left[x^{10} + x^8 + x^5 + x^4 + x^2 + x + 1\right]$. Parameters $\left[15,5,7\right]_2$.
      \item  Generator polynomial $\left[\mathcal{I}^{(1)}_{\perp}(\mathscr{C})\right]= \left[x + 1\right]$. Parameters $\left[15,14,2\right]_2$.
 \end{itemize}

 It follows directly that the lower and upper bounds on the minimum symplectic distance of $\mathscr{C}^{\perp_s}$ are both $4$ based on the parameters of the relevant cyclic code, so the parameters of $\mathscr{C}^{\perp_s}$ are $\left[30,19,4\right]_2^s$. By using Magma \cite{bosma1997magma}, we can compute that the specific parameters of $\mathscr{C}^{\perp_s}$ are indeed $\left[30,19,4\right]_2^s$.
Thus, by Theorem \ref{CRSS}, we can obtain a binary QECC with parameters $\left[\left[15,4,4\right]\right]_2$, which is an optimal QECC according to \cite{Grassltable}.
\end{example}


By utilizing the results above, we devise many superior symplectic self-orthogonal codes via computerized search, corresponding to 117 record-breaking QECCs, see Table \ref{all_QECCs}.
Specifically, we give the specific generators of related codes and their corresponding parameters in Tables \ref{F_QECCs} and \ref{T_QECCs}. 
In order to save space, we give them in abbreviated form and represent the coefficient polynomials in ascending order, where the indexes of elements represent consecutive elements of the same number. For example, polynomial $1 + x^2 + x^3 + x^4$ over $\mathbb{F}_2$ is denoted as $101^3$. The specific parameters of corresponding codes are calculated by using Magma \cite{bosma1997magma}.


\begin{remark}
	The new binary QECCs presented in this paper have been updated in Grassl's code table \cite{Grassltable}.  Additionally, after we completed this paper, Grassl et al. \cite{ezerman2024characterization} systematically constructed a large number of good QECCs using QC codes, and extended the quantum code table to $q\le 8$ and $n\le 100$.
\end{remark}

\section{Conclusion}\label{sec5}

In this work, we established the necessary and sufficient conditions for QC codes with even 
index to be symplectic self-orthogonal. 
Furthermore, through the decomposition of the code space of QC codes, we established lower and upper bounds on the minimum symplectic distances of a class of $1$-generator QC codes with index two and their symplectic dual codes. 
This contributed to improving the results presented in \cite{galindo2018quasi}, \cite{dastbasteh2023polynomial}, and \cite{GuanAdditive2023}, see Remarks \ref{improve_bound} and \ref{improve_bound_Galindo}, and Examples \ref{good_Ex1} and \ref{good_Ex_for additive}.
As an application, we constructed 117 record-breaking binary QECCs from symplectic self-orthogonal QC codes.

\begin{table*}[!t]
	\centering
	\begin{threeparttable}
		\caption{ New binary QECCs  and derived record-breaking codes\tnote{1}$^{,}$\tnote{2}}\label{all_QECCs}
		\footnotesize
		\begin{tabular}{cccccc}
			\hline
			NO. & Record-breaking  QECCs  & Code Table \cite{Grassltable} & NO. &  Record-breaking QECCs  & Code Table \cite{Grassltable} \\ \hline
			 1  & $[[40,6,10]]_2^{\star}$ &       $[[40,6,8]]_2^*$        & 60  & $[[63,31,8]]_2^{\star}$ &               -               \\
			 2  &     $[[41,6,10]]_2$     &        $[[41,6,8]]_2$         & 61  &   $[[73,18,13]]_2^*$    &       $[[73,18,12]]_2$        \\
			 3  &     $[[39,7,9]]_2$      &        $[[39,7,8]]_2$         & 62  &    $[[73,17,13]]_2$     &       $[[73,17,12]]_2$        \\
			 4  &     $[[39,6,9]]_2$      &        $[[39,6,8]]_2$         & 63  &    $[[73,16,13]]_2$     &       $[[73,16,12]]_2$        \\
			 5  &     $[[40,7,9]]_2$      &        $[[40,7,8]]_2$         & 64  &    $[[73,15,13]]_2$     &       $[[73,15,12]]_2$        \\
			 6  & $[[42,7,10]]_2^{\star}$ &        $[[42,7,9]]_2$         & 65  &    $[[73,14,13]]_2$     &       $[[73,14,12]]_2$        \\
			 7  &     $[[43,7,10]]_2$     &        $[[43,7,9]]_2$         & 66  &    $[[74,18,13]]_2$     &       $[[74,18,12]]_2$        \\
			 8  &     $[[44,7,10]]_2$     &        $[[44,7,9]]_2$         & 67  &    $[[74,17,13]]_2$     &       $[[74,17,12]]_2$        \\
			 9  &     $[[45,7,10]]_2$     &        $[[45,7,9]]_2$         & 68  &    $[[74,16,13]]_2$     &       $[[74,16,12]]_2$        \\
			10  &     $[[46,7,10]]_2$     &        $[[46,7,9]]_2$         & 69  &    $[[74,15,13]]_2$     &       $[[74,15,12]]_2$        \\
			11  &     $[[41,8,9]]_2$      &        $[[41,8,8]]_2$         & 70  &    $[[75,18,13]]_2$     &       $[[75,18,12]]_2$        \\
			12  &    $[[40,5,10]]_2^*$    &        $[[40,5,9]]_2$         & 71  &    $[[75,17,13]]_2$     &       $[[75,17,12]]_2$        \\
			13  &     $[[41,5,10]]_2$     &        $[[41,5,9]]_2$         & 72  &    $[[75,16,13]]_2$     &       $[[75,16,12]]_2$        \\
			14  &     $[[42,5,10]]_2$     &        $[[42,5,9]]_3$         & 73  &    $[[75,15,13]]_2$     &       $[[75,15,12]]_2$        \\
			15  &     $[[39,6,9]]_2$      &        $[[39,6,8]]_4$         & 74  &    $[[76,18,13]]_2$     &       $[[76,18,12]]_2$        \\
			16  &     $[[40,6,9]]_2$      &        $[[40,6,8]]_5$         & 75  &    $[[76,17,13]]_2$     &       $[[76,17,12]]_2$        \\
			17  &     $[[41,6,9]]_2$      &        $[[41,6,8]]_6$         & 76  &    $[[76,16,13]]_2$     &       $[[76,16,12]]_2$        \\
			18  &    $[[42,6,10]]_2^*$    &        $[[42,6,9]]_2$         & 77  &    $[[76,15,13]]_2$     &       $[[76,15,12]]_2$        \\
			19  &     $[[43,6,10]]_2$     &        $[[43,6,9]]_2$         & 78  &    $[[77,18,13]]_2$     &       $[[77,18,12]]_2$        \\
			20  &     $[[44,6,10]]_2$     &        $[[44,6,9]]_2$         & 79  &    $[[77,17,13]]_2$     &       $[[77,17,12]]_2$        \\
			21  &     $[[45,6,10]]_2$     &        $[[45,6,9]]_2$         & 80  &    $[[77,16,13]]_2$     &       $[[77,16,12]]_2$        \\
			22  &     $[[41,7,9]]_2$      &        $[[41,7,8]]_2$         & 81  &    $[[77,15,13]]_2$     &       $[[77,15,12]]_2$        \\
			23  &     $[[41,6,9]]_2$      &        $[[41,6,8]]_2$         & 82  &    $[[78,18,13]]_2$     &       $[[78,18,12]]_2$        \\
			24  &    $[[42,13,8]]_2^*$    &        $[[42,13,7]]_2$        & 83  &    $[[78,17,13]]_2$     &       $[[78,17,12]]_2$        \\
			25  &    $[[42,14,8]]_2^*$    &        $[[42,14,7]]_2$        & 84  &    $[[78,16,13]]_2$     &       $[[78,16,12]]_2$        \\
			26  &    $[[45,4,11]]_2^*$    &        $[[45,4,10]]_2$        & 85  &    $[[78,15,13]]_2$     &       $[[78,15,12]]_2$        \\
			27  &     $[[46,4,11]]_2$     &        $[[46,4,10]]_2$        & 86  &    $[[79,18,13]]_2$     &       $[[79,18,12]]_2$        \\
			28  &     $[[47,4,11]]_2$     &        $[[47,4,10]]_2$        & 87  &    $[[79,17,13]]_2$     &       $[[79,17,12]]_2$        \\
			29  &    $[[45,5,11]]_2^*$    &        $[[45,5,10]]_2$        & 88  &    $[[79,16,13]]_2$     &       $[[79,16,12]]_2$        \\
			30  &     $[[46,5,11]]_2$     &        $[[46,5,10]]_2$        & 89  &    $[[80,18,13]]_2$     &       $[[80,18,12]]_2$        \\
			31  &     $[[47,5,11]]_2$     &        $[[47,5,10]]_2$        & 90  &    $[[80,17,13]]_2$     &       $[[80,17,12]]_2$        \\
			32  &     $[[48,5,11]]_2$     &        $[[48,5,10]]_2$        & 91  &    $[[80,16,13]]_2$     &       $[[80,16,12]]_2$        \\
			33  &    $[[45,6,10]]_2^*$    &        $[[45,6,9]]_2$         & 92  &    $[[81,18,13]]_2$     &       $[[81,18,12]]_2$        \\
			34  &    $[[45,10,9]]_2^*$    &        $[[45,10,8]]_2$        & 93  &    $[[72,19,12]]_2$     &       $[[72,19,11]]_2$        \\
			35  &    $[[48,3,12]]_2^*$    &        $[[48,3,11]]_2$        & 94  &    $[[72,18,12]]_2$     &       $[[72,18,11]]_2$        \\
			36  & $[[45,21,7]]_2^{\star}$ &        $[[45,21,6]]_2$        & 95  &   $[[78,25,11]]_2^*$    &       $[[78,25,10]]_2$        \\
			37  &    $[[48,5,11]]_2^*$    &        $[[48,5,10]]_2$        & 96  &   $[[78,25,12]]_2^*$    &       $[[78,25,10]]_2$        \\
			38  &    $[[49,4,12]]_2^*$    &        $[[49,4,11]]_2$        & 97  &    $[[78,24,12]]_2$     &       $[[78,24,11]]_2$        \\
			39  &    $[[51,8,11]]_2^*$    &        $[[51,8,10]]_2$        & 98  &    $[[78,23,12]]_2$     &       $[[78,23,11]]_2$        \\
			40  &     $[[51,7,11]]_2$     &        $[[51,7,10]]_2$        & 99  &    $[[78,22,12]]_2$     &       $[[78,22,11]]_2$        \\
			41  &     $[[51,6,11]]_2$     &        $[[51,6,10]]_2$        & 100 &    $[[78,21,12]]_2$     &       $[[78,21,11]]_2$        \\
			42  &     $[[50,9,10]]_2$     &        $[[50,9,9]]_2$         & 101 &    $[[79,25,12]]_2$     &       $[[79,25,11]]_2$        \\
			43  &    $[[51,9,11]]_2^*$    &        $[[51,9,10]]_2$        & 102 &    $[[79,24,12]]_2$     &       $[[79,24,11]]_2$        \\
			44  &    $[[50,10,10]]_2$     &        $[[50,10,9]]_2$        & 103 &    $[[79,23,12]]_2$     &       $[[79,23,11]]_2$        \\
			45  & $[[55,20,9]]_2^{\star}$ &               -               & 104 &    $[[79,22,12]]_2$     &       $[[79,22,11]]_2$        \\
			46  &   $[[56,11,11]]_2^*$    &       $[[56,11,10]]_2$        & 105 &    $[[80,25,12]]_2$     &       $[[80,25,11]]_2$        \\
			47  &    $[[57,11,11]]_2$     &       $[[57,11,10]]_2$        & 106 &    $[[80,24,12]]_2$     &       $[[80,24,11]]_2$        \\
			48  &    $[[58,11,11]]_2$     &       $[[58,11,10]]_2$        & 107 &    $[[80,23,12]]_2$     &       $[[80,23,11]]_2$        \\
			49  &    $[[60,6,13]]_2^*$    &        $[[60,6,12]]_2$        & 108 &    $[[81,25,12]]_2$     &       $[[81,25,11]]_2$        \\
			50  &   $[[60,10,12]]_2^*$    &       $[[60,10,11]]_2$        & 109 &    $[[81,24,12]]_2$     &       $[[81,24,11]]_2$        \\
			51  &     $[[60,9,12]]_2$     &        $[[60,9,11]]_2$        & 110 &    $[[82,25,12]]_2$     &       $[[82,25,11]]_2$        \\
			52  &    $[[61,10,12]]_2$     &       $[[61,10,11]]_2$        & 111 &    $[[77,26,11]]_2$     &       $[[77,26,10]]_2$        \\
			53  &   $[[60,11,12]]_2^*$    &       $[[60,11,11]]_2$        & 112 &    $[[77,25,11]]_2$     &       $[[77,25,10]]_2$        \\
			54  &    $[[61,11,12]]_2$     &       $[[61,11,11]]_2$        & 113 &    $[[77,24,11]]_2$     &       $[[77,24,10]]_2$        \\
			55  &    $[[59,12,11]]_2$     &       $[[59,12,10]]_2$        & 114 &   $[[78,26,11]]_2^*$    &       $[[78,26,10]]_2$        \\
			56  &    $[[62,6,14]]_2^*$    &        $[[62,6,13]]_2$        & 115 &    $[[79,26,11]]_2$     &       $[[79,26,10]]_2$        \\
			57  &     $[[62,5,14]]_2$     &        $[[62,5,13]]_2$        & 116 &   $[[78,29,10]]_2^*$    &        $[[78,29,9]]_2$        \\
			58  &     $[[61,7,13]]_2$     &        $[[61,7,12]]_2$        & 117 &   $[[105,80,6]]_2^*$    &       $[[105,80,5]]_2$        \\
			59  &     $[[61,6,13]]_2$     &        $[[61,6,12]]_2$        &  -  &            -            &               -               \\ \hline
		\end{tabular}
		\begin{tablenotes}    
			\footnotesize               
			\item[1] The QECCs with superscripts ``$*$" and ``$\star$" can be found in Table \ref{F_QECCs} and Table \ref{T_QECCs} as construction polynomials, respectively. The rest of the QECCs are derived by Lemma \ref{propagation_rules}.          
			\item[2] Some additions: 
			$[[55,20,9]]_2$ and $[[63,31,8]]_2$ were previously presented in our first version \cite{guan2021new}, but were not formally published in \cite{guan2022new}; here we formally submit them for publication.
			After completing the first version of this paper \cite{guan2022symplectic}, we note that in the very near time, Dastbasteh et al. \cite{dastbasteh2023polynomial} also obtained $[[45,6,10]]_2$, $[[45,10,9]]_2$ and $[[51,8,11]]_2$, and Seneviratne \cite{Grassltable} obtained $[[40,6,10]]_2$, respectively.
		\end{tablenotes}            
	\end{threeparttable}       
\end{table*}

\begin{table*}[!t]
	\caption{One-generator symplectic self-orthogonal QC codes with $\ell=2$ and record-breaking  QECCs}\label{F_QECCs}
	\centering
	\footnotesize
	\begin{tabular}{ccccc}
		\hline
		No. &Codes      &Symplectic Dual   &                                                                                                                                                                                                                             $g(x),f_0(x),f_1(x)$                                                                                                                                                                                                                              &    Record-breaking  QECCs     \\ \hline
		 1  &  $[80,35]_2^s$  &  $[80,45,10]_2^s$   &                                                                                                                     \makecell[c]{$1^{2}0^{2}1^2$, $0^{4}1^{3}01^{3}0101^{3}0^{2}10^{2}1^{3}0101^{3}01^3$,\\
			$0101^{2}0^{3}1^{2}0^{5}10^{2}1010^{2}10^{5}1^{2}0^{3}1^{2}01$}                                                                                                                      & $[[40,5,10]]_2$  \\
		 2  &  $[84,36]_2^s$  &  $[84,48,10]_2^s$   &                                                                                                                                                \makecell[c]{$101^2,1^{3}0^{3}1^{2}01^{3}0101^{2}0^{3}1^{3}0^{3}1^{2}0101^{3}01^{2}0^{3}1^2$,\\$010^{2}10^{2}1^{2}0^{3}10101^{2}0^{2}1010^{2}1^{2}01010^{3}1^{2}0^{2}10^{2}1$}                                                                                                                                                 & $[[42,6,10]]_2$  \\
		 3  &  $[84,29]_2^s$  &   $[84,55,8]_2^s$   &                                                                                                                                                     \makecell[c]{$101010^{3}10^{3}1,1^{2}0^{2}101^{2}0^{2}1^{3}01^{2}01^{2}01^{3}01^{2}01^{2}01^{3}0^{2}1^{2}010^{2}1$,\\$0^{7}10^{3}1^{2}0^{4}10^{7}10^{4}1^{2}0^{3}1$}                                                                                                                                                      & $[[42,13,8]]_2$  \\
		 4  &  $[84,28]_2^s$  &   $[84,56,8]_2^s$   &                                                                                                                                                         \makecell[c]{$10^{3}10^{3}10101,0^{4}10^{3}10^{3}101^{3}0^{2}10^{3}10^{2}1^{3}010^{3}10^{3}1$,\\$01^{3}0101^{6}0^{3}10^{2}1^{2}01^{2}0^{2}10^{3}1^{6}0101^3$}                                                                                                                                                         & $[[42,14,8]]_2$  \\
		 5  &  $[90,41]_2^s$  &  $[90,49,11]_2^s$   &                                                                                                                                                        \makecell[c]{$10^{2}1^2,1010101^{3}0^{2}10^{2}1^{3}01^{10}01^{3}0^{2}10^{2}1^{3}0101$,\\$101^{4}0^{3}1^{2}01010^{2}101^{3}0^{2}1^{3}010^{2}10101^{2}0^{3}1^4$ }                                                                                                                                                        & $[[45,4,11]]_2$  \\
		 6  &  $[90,40]_2^s$  &  $[90,50,11]_2^s$   &                                                                                                                                                  \makecell[c]{$10101^2,1^{2}0^{5}1010101^{2}0^{2}1^{2}0^{8}1^{2}0^{2}1^{2}0101010^{5}1$,\\$1^{2}0^{2}1^{3}0^{2}101^{3}010^{3}1^{3}0^{2}1^{3}0^{3}101^{3}010^{2}1^{3}0^{2}1$}                                                                                                                                                  & $[[45,5,11]]_2$  \\
		 7  &  $[90,39]_2^s$  &  $[90,51,10]_2^s$   &                                                                                                                                                   \makecell[c]{$10^{2}10^{2}1,10101^{2}0^{3}10^{2}10^{3}10^{2}10^{6}10^{2}10^{3}10^{2}10^{3}1^{2}01$,\\$1^{2}01^{3}0^{2}1010^{6}1010^{2}1^{2}0^{2}1010^{6}1010^{2}1^{3}01$}                                                                                                                                                   & $[[45,6,10]]_2$  \\
		 8  &  $[90,35]_2^s$  &   $[90,55,9]_2^s$   &                                                                                                                                                                                $1^5,1^{2}01^{2}010^{3}1^{4}010101^{3}0^{2}1^{3}010101^{4}0^{3}101^{2}01,10^{5}10^{5}10^{5}1^{10}0^{5}10^{5}1$                                                                                                                                                                                 & $[[45,10,9]]_2$  \\
		 9  &  $[96,45]_2^s$  &  $[96,51,12]_2^s$   &                                                                                                                                                      \makecell[c]{$1^2,010^{2}10^{9}1^{2}01010^{2}1^{2}01^{2}0^{2}10101^{2}0^{9}10^{2}1$,\\$01^{2}010^{2}101^{5}0^{2}1^{3}01^{2}0^{5}1^{2}01^{3}0^{2}1^{5}010^{2}101^2$}                                                                                                                                                      & $[[48,3,12]]_2$  \\
		 10  &  $[96,43]_2^s$  &  $[96, 53,11]_2^s$  &                                                                                                                                                       \makecell[c]{$1^2,0^{4}1^{3}0101^{3}01^{7}01^{2}01^{2}01^{7}01^{3}0101^3$,\\$0^{4}10^{2}10^{2}1^{3}01^{2}01010^{2}1^{2}01^{2}0^{2}10101^{2}01^{3}0^{2}10^{2}1$  }                                                                                                                                                       & $[[48,5,11]]_2$  \\
		11  & $[98, 45]_2^s$  &  $[98, 53,12]_2^s$  &                                                                                                                                       \makecell[c]{$101^2,010^{2}1^{2}0^{2}1^{2}010^{2}10^{6}1^{2}0^{4}1^{2}0^{6}10^{2}101^{2}0^{2}1^{2}0^{2}1$,\\$0^{3}1^{2}0^{4}10^{2}101^{2}01010^{2}1^{2}0^{2}1^{2}0^{2}10101^{2}010^{2}10^{4}1^2$}                                                                                                                                       & $[[49,4,12]]_2$  \\
		12  & $[102, 43]_2^s$ & $[102, 59,11]_2^s$  &                                                                                                                                     \makecell[c]{$10^{3}1^{2}01^2,1^{2}0101^{3}01^{2}0^{4}10^{3}1^{2}0101^{4}0101^{2}0^{3}10^{4}1^{2}01^{3}0101$,\\$1^{3}01^{6}0^{2}1^{2}010^{3}1^{2}0^{2}101^{2}010^{2}1^{2}0^{3}101^{2}0^{2}1^{6}01^2$}                                                                                                                                     & $[[51,8,11]]_2$  \\
		13  & $[102, 42]_2^s$ & $[102, 60,11]_2^s$  &                                                                                                                                            \makecell[c]{$10^{3}1^{2}01^2,0^{2}1010^{3}1^{7}01^{2}01^{3}0^{2}10^{2}10^{2}1^{3}01^{2}01^{7}0^{3}101$,\\$01^{5}0^{2}10^{4}1^{2}0^{4}10101^{2}0^{2}1^{2}01010^{4}1^{2}0^{4}10^{2}1^5$}                                                                                                                                            & $[[51,9,11]]_2$  \\
		14  & $[112, 45]_2^s$ & $[112, 67,11]_2^s$  &                                                                                                                                   \makecell[c]{$1^{2}0^{2}1^6,101010^{3}101^{2}0^{3}1^{4}0101^{2}0^{3}1^{3}0^{3}1^{2}0101^{4}0^{3}1^{2}010^{3}101$,\\$1^{4}0^{2}10^{9}1^{2}0^{2}1^{2}01^{3}0^{2}10^{2}1^{3}01^{2}0^{2}1^{2}0^{9}10^{2}1^3$}                                                                                                                                   & $[[56,11,11]]_2$ \\
		15  & $[120, 54]_2^s$ & $[120, 66,13]_2^s$  &                                                                                                                                         \makecell[c]{$1^{2}0^{2}1,0101010^{3}1^{5}01^{3}01^{2}0^{2}10^{3}1^{2}0^{3}1^{2}0^{3}10^{2}1^{2}01^{3}01^{5}0^{3}10101$,\\$1^{5}01^{4}0^{5}1010^{2}1^{5}01^{9}01^{5}0^{2}1010^{5}1^{4}01^4$}                                                                                                                                          & $[[60,6,13]]_2$  \\
		16  & $[120, 50]_2^s$ & $[120, 70,12]_2^s$  &                                                                                                                         \makecell[c]{$10^{5}101,0^{2}10^{3}10^{2}1^{2}01^{4}010^{4}1^{2}0^{2}1^{2}0^{5}1^{2}0^{2}1^{2}0^{4}101^{4}01^{2}0^{2}10^{3}1$,\\$01^{3}01^{3}01^{2}0^{3}101^{3}010^{2}1^{3}010^{5}101^{3}0^{2}101^{3}010^{3}1^{2}01^{3}01^3$}                                                                                                                         & $[[60,10,12]]_2$ \\
		17  & $[120, 49]_2^s$ & $[120, 71,12]_2^s$  &                                                                                                                           \makecell[c]{$10^{5}101,1^{2}010^{3}1^{4}0^{2}1010^{2}1010^{3}101^{3}0101^{3}010^{3}1010^{2}1010^{2}1^{4}0^{3}101$,\\$0^{3}10^{3}1^{3}01^{3}0^{3}1^{2}0^{2}10^{3}1^{2}0^{7}1^{2}0^{3}10^{2}1^{2}0^{3}1^{3}01^{3}0^{3}1$}                                                                                                                            & $[[60,11,12]]_2$ \\
		18  & $[124, 56]_2^s$ & $[124, 68,14]_2^s$  &                                                                                                                        \makecell[c]{$1^{4}01,1^{2}0^{2}101^{4}0^{7}1^{3}0^{4}10^{2}10^{3}10^{3}10^{2}10^{4}1^{3}0^{7}1^{4}010^{2}1$,\\$10^{2}1^{2}0^{3}1^{2}01^{4}010^{2}1010^{2}1^{2}01^{2}0^{2}10^{2}1^{2}01^{2}0^{2}1010^{2}101^{4}01^{2}0^{3}1^2$}                                                                                                                        & $[[62,6,14]]_2$  \\
		19  & $[146, 55]_2^s$ & $[146, 91,13]_2^s$  &                                                                                                     \makecell[c]{$1^{2}0^{3}10^{2}1^{3}0^{4}1^{2}01$,\\$10101^{3}0^{7}1^{2}010^{3}10^{2}1^{2}0^{2}1^{6}010^{2}101^{6}0^{2}1^{2}0^{2}10^{3}101^{2}0^{7}1^{3}01$,\\$1^{6}0^{2}101^{3}0^{2}10^{2}10^{2}1^{2}0^{6}1^{4}0^{8}1^{4}0^{6}1^{2}0^{2}10^{2}10^{2}1^{3}010^{2}1^5$}                                                                                                     & $[[73,18,13]]_2$ \\
		20  & $[156, 53]_2^s$ & $[156, 103,11]_2^s$ &                                                                                         \makecell[c]{$10101010^{7}10101010^{3}1$,\\$101^{2}0^{3}1^{2}0^{5}10^{3}1010^{2}10^{3}1^{3}010^{2}10^{4}10^{4}10^{2}101^{3}0^{3}10^{2}1010^{3}10^{5}1^{2}0^{3}1^2$,\\$01010^{2}101^{2}0^{5}10^{2}101^{4}0^{4}10^{2}101^{3}01^{2}01^{2}01^{3}010^{2}10^{4}1^{4}010^{2}10^{5}1^{2}010^{2}101$}                                                                                          & $[[78,25,11]]_2$ \\
		21  & $[156, 53]_2^s$ & $[156, 103,12]_2^s$ &                                                                                                 \makecell[c]{$10101010^{7}10101010^{3}1$,\\$10^{2}101^{2}01^{6}0101^{5}0^{2}1^{2}01^{2}01^{7}0^{2}10^{2}1^{7}01^{2}01^{2}0^{2}1^{5}0101^{6}01^{2}01$,\\$1^{5}0^{2}10^{3}1010^{4}1^{2}0^{2}1^{4}0^{5}10^{2}1^{2}01^{5}01^{2}0^{2}10^{5}1^{4}0^{2}1^{2}0^{4}1010^{3}10^{2}1^4$}                                                                                                 & $[[78,25,12]]_2$ \\
		22  & $[156, 52]_2^s$ & $[156, 104,11]_2^s$ &                                                                                                     \makecell[c]{$10101010^{7}10101010^{3}1$,\\$0^{3}101^{3}0^{5}10^{4}1010^{2}10^{4}1^{2}01010^{2}1^{3}01^{3}0^{2}10101^{2}0^{4}10^{2}1010^{4}10^{5}1^{3}01$,\\$01^{9}01^{2}0^{3}10101^{2}0^{3}10^{2}10^{2}101^{2}0^{9}1^{2}010^{2}10^{2}10^{3}1^{2}01010^{3}1^{2}01^9$}                                                                                                     & $[[78,26,11]]_2$ \\
		23  & $[156, 49]_2^s$ & $[156, 107,10]_2^s$ &                                                                                \makecell[c]{$10101010^{7}10101010^{3}1$,\\$10^{3}10101^{2}0^{5}1^{2}010^{2}10^{4}1010^{2}1^{2}0^{5}1^{3}0^{5}1^{2}0^{2}1010^{4}10^{2}101^{2}0^{5}1^{2}0101$,\\ $0^{2}101^{2}010^{2}1^{2}0^{5}1^{3}01^{2}01^{2}01^{2}0^{2}10^{2}10^{2}10^{3}10^{2}10^{2}10^{2}1^{2}01^{2}01^{2}01^{3}0^{5}1^{2}0^{2}101^{2}01$}                                                                                & $[[78,29,10]]_2$ \\
		24  & $[210, 25]_2^s$ & $[210, 185,6]_2^s$  & \makecell[c]{$10^{2}1^{3}0^{2}1^{4}010^{5}1^{4}0^{2}1^{4}0^{3}10^{4}1^{4}01^{3}01^{2}01010^{5}1^{2}010^{4}101^{3}0^{5}1^{2}0^{2}1^2$,\\$10^{2}1^{2}0^{2}101^{5}0^{4}101^{2}0^{2}1^{2}01^{4}01^{3}0101^{3}0^{2}1^{2}0^{2}1^{2}0^{2}1^{2}0^{2}$\\$10^{2}10^{2}1^{4}010^{7}1^{2}0101^{3}01010^{5}1^{2}01^{2}0^{2}1010^{3}1,10^{2}1^{2}01^{2}0101^{2}01^{5}$\\$0^{3}1^{2}0^{2}10^{6}1^{4}0^{4}1010101^{3}0101^{3}0^{2}101^{2}010^{2}10^{7}1^{4}0^{4}1^{3}01^{4}0^{2}10^{3}101^2$} & $[[105,80,6]]_2$ \\ \hline
	\end{tabular}
\end{table*}

\begin{table*}[!t]
			\centering
		\begin{threeparttable}
	\caption{Two-generator symplectic self-orthogonal QC codes\tnote{3} with $\ell=2$  and record-breaking  QECCs}\label{T_QECCs}
	\centering
	\footnotesize
	\begin{tabular}{ccccc}	
		\hline
		No. &     Codes      &   Symplectic Dual  &                                                                                                                                $g_1(x),g_2(x),f(x)$                                                                                                                                 &    Record-breaking  QECCs    \\ \hline
		 1  & $[80,34]_2^s$  &  $[80,46,10]_2^s$   &                                                                                               \makecell[c]{$1^8$, $1^{40}$, $0^{4}1^{5}01^{2}0^{2}101^{3}0^{3}1^{3}010^{2}1^{2}01^5$}                                                                                               & $[[40,6,10]]_2$ \\
		 2  & $[84,35]_2^s$  &  $[84,49,10]_2^s$   &                                                                                               \makecell[c]{$10^{7}101$, $101^{3}0^{2}101^{3}0^{2}101^{3}0^{2}101^{3}0^{2}101^{3}0^{2}101^3$,\\
		 	$1^{3}0^{8}1^{3}0^{4}1^{7}0^{4}1^{3}0^{8}1^2$}                                                                                               & $[[42,7,10]]_2$ \\
		 3  & $[90,24]_2^s$  & $[ 90, 66, 7 ]_2^s$ &                                          \makecell[c]{$1^{4}0101^{4}0^{2}101^{2}0^{2}10^{3}1$,\\
	$1^{2}01^{2}01^{2}01^{2}01^{2}01^{2}01^{2}01
	^{2}01^{2}01^{2}01^{2}01^{2}01^{2}01^{2}01^2$,\\
	$0101^{2}0^{5}1010^{20}1010^{5}1^{2}01$}                                          & $[[45,21,7]]_2$ \\
		 4  & $[110,35]_2^s$ &  $[110,75,9]_2^s$   &                                                  \makecell[c]{$10^{2}101^{2}010^{2}10^{3}1^{2}0^{4}10^{2}10^{2}1010101^2$,\\
	$1010^{2}1010^{2}1^{4}01^{4}01^{2}01^{4}01^{4}0^{2}1010^{2}101$,\\
	$010^{3}1^{2}0^{42}1^{2}0^{3}1$}                                                  & $[[55,20,9]]_2$ \\
		 5  & $[126,32]_2^s$ &  $[126,94,8]_2^s$   & \makecell[c]{$10^{4}101^{2}0^{2}1^{2}0^{4}10^{2}10101010101^{2}01^{3}01^3$,\\
	$1010^{2}1^{2}0^{2}1^{2}010^{2}1^{3}0^{2}10^{2}1010^{3}1^{2}0101^{2}01010^{4}1^{3}0^{4}1^{2}0^{2}1$,\\ $0^{2}1^{2}0^{2}1^{3}010^{5}101^{2}010101^{2}010^{6}101^{2}010101^{2}010^{5}101^{3}0^{2}1^2$} & $[[63,31,8]]_2$ \\ \hline
	\end{tabular}
\begin{tablenotes}    
	\footnotesize               
	\item[3] These $2$-generator QC codes with generators $([g_1(x)f(x)],[g_1(x)])$ and $([g_2(x)],[g_2(x)f(x)])$.        
\end{tablenotes}            
\end{threeparttable}       
\end{table*}


\bibliographystyle{IEEEtran} 
\bibliography{reference} 
\end{document}